\documentclass[a4paper,11pt]{article}

\usepackage{amsmath}
\usepackage{amssymb}
\usepackage{amsthm}
\usepackage[backend=bibtex,style=numeric,giveninits=true,maxbibnames=99]{biblatex}
\usepackage[labelfont=bf,labelsep=period]{caption}
\usepackage{comment}
\usepackage[nodayofweek]{datetime}
\usepackage{enumitem}
\usepackage{float}
\usepackage[margin=2.5cm]{geometry}
\usepackage{graphicx}
\usepackage{hyperref}
\usepackage[utf8]{inputenc}
\usepackage{lipsum}
\usepackage{multicol}
\usepackage{newtx}
\usepackage{numbertabbing} %
\usepackage{thm-restate}
\usepackage{url}
\usepackage{xcolor}
\usepackage{xspace}
\usepackage{tikz}

\renewbibmacro*{doi+eprint+url}{%
  \iftoggle{bbx:url}
  {\iffieldundef{doi}{\usebibmacro{url}}{}}
  {}%
  \newunit\newblock
  \iftoggle{bbx:eprint}
  {\usebibmacro{eprint}}
  {}%
  \newunit\newblock
  \iftoggle{bbx:doi}
  {\printfield{doi}}
  {}
}

\allowdisplaybreaks[2] %

\newtheoremstyle{plain-boldhead}%
  {\topsep}%
  {\topsep}%
  {\itshape}%
  {}%
  {\bfseries}%
  {.}%
  { }%
  {\thmname{#1}\thmnumber{ #2}\thmnote{ (\bfseries #3)}}%
\newtheoremstyle{definition-boldhead}%
  {\topsep}%
  {\topsep}%
  {\normalfont}%
  {}%
  {\bfseries}%
  {.}%
  { }%
  {\thmname{#1}\thmnumber{ #2}\thmnote{ (\bfseries #3)}}%
\theoremstyle{plain-boldhead}
\newtheorem{theorem}{Theorem}
\newtheorem{proposition}[theorem]{Proposition}
\newtheorem{lemma}[theorem]{Lemma}
\newtheorem{corollary}[theorem]{Corollary}

\theoremstyle{definition-boldhead}
\newtheorem{definition}{Definition}

\newtheorem{example}{Example}

\newdateformat{simple}{\THEDAY\ \monthname[\THEMONTH]\ \THEYEAR}
\simple

\floatstyle{ruled}
\newfloat{algo}{htbp}{algo}
\floatname{algo}{Algorithm}

\def \ifempty#1{\def\temp{#1} \ifx\temp\empty }

\newcommand{\var}[1]{\textit{#1}}
\newcommand{\op}[1]{\textsl{#1}}

\newcommand{\strf}[1]{\textsf{#1}}

\newcommand{\false}{\textsc{false}\xspace}
\newcommand{\true}{\textsc{true}\xspace}

\newcommand{\BF}{\ensuremath{\mathbb{F}}\xspace}

\newcommand{\BN}{\ensuremath{\mathbb{N}}\xspace}

\newcommand{\BQ}{\ensuremath{\mathbb{Q}}\xspace}
\newcommand{\BR}{\ensuremath{\mathbb{R}}\xspace}

\newcommand{\CA}{\ensuremath{\mathcal{A}}\xspace}

\newcommand{\CC}{\ensuremath{\mathcal{C}}\xspace}

\newcommand{\CF}{\ensuremath{\mathcal{F}}\xspace}

\newcommand{\CK}{\ensuremath{\mathcal{K}}\xspace}

\newcommand{\CO}{\ensuremath{\mathcal{O}}\xspace}
\newcommand{\CP}{\ensuremath{\mathcal{P}}\xspace}
\newcommand{\CQ}{\ensuremath{\mathcal{Q}}\xspace}
\newcommand{\CR}{\ensuremath{\mathcal{R}}\xspace}

\newcommand{\reliable}{reliable\xspace}

\newcommand\node{\var{node}\xspace}

\def \ifempty#1{\def\temp{#1} \ifx\temp\empty }

\begin{document}

\title{\bf Monotone Erasure Codes}

\author{Vivien Bammert\footnote{%
    Contact authors: Institute of Computer Science, University of Bern,
    Neubr\"{u}ckstrasse 10, 3012 CH-Bern, Switzerland.}\\
  University of Bern\\
  \url{vivien.bammert@unibe.ch}
  \and Annalisa Cimatti\footnotemark[1]\\
  University of Bern\\
  \url{annalisa.cimatti@unibe.ch}
  \and Orestis Alpos\\
  Common Prefix\\
  \url{orapos@gmail.com}
  \and Giuliano Losa\\
  Stellar Development Foundation\\
  \url{giuliano@stellar.org}\\
  \and Christian Cachin\\
  University of Bern\\
  \url{christian.cachin@unibe.ch}
}

\date{\today}

\maketitle

\begin{abstract}\noindent
Erasure codes are a critical component in reliable storage systems today, and
many blockchain systems use consensus protocols that involve erasure codes to
reduce their communication cost. Existing erasure codes, such as Reed-Solomon
codes, make it easy to design systems that rely on a threshold failure
assumption, meaning that at most a fixed threshold of nodes may fail, regardless
of which nodes fail. However, recent blockchain systems have departed from this
simple model.  Instead, they use generalized failure assumptions or general
Byzantine quorum systems that allow expressing that, e.g., some nodes may be
more trustworthy than others.

This paper introduces \emph{monotone erasure codes} that respect arbitrary
trust assumptions on a set of nodes, characterized by a monotone
access structure or the related notion of a quorum system.  The paper first
describes a method for constructing a monotone erasure code from any access
structure given by a monotone Boolean formula (MBF).  Next, the important
notion of a \emph{linear monotone erasure code} is introduced, which works
on vectors over a finite field and where the encoding is a linear operation.
A method to build a linear monotone erasure  code for any access structure is presented. This method is  computationally efficient, but it does not provide an optimal code in general.
An alternative algorithm to construct an optimal linear monotone erasure code
for an access structure given as an MBF is shown that uses an ordinary,
maximum-distance separable erasure code underneath.  An important class of
non-threshold access structures used in the practice of cryptocurrencies is
then examined, and an optimal monotone erasure code for these access
structures is constructed. 

Last but not least, this work also shows how to use monotone erasure codes to
obtain a communication-efficient, generalized version of the well-known
asynchronous verifiable information dispersal (AVID) primitive, which is a key
building block for developing efficient reliable broadcast and consensus protocols.
\end{abstract}

\section{Introduction}
\label{sec:introduction}

As computing systems collect and process more information than ever today, it
becomes critical to store and communicate tremendous amounts of information
securely, reliably, and economically.
An \emph{erasure code} provides a well-known
method to divide large volumes of data into pieces so that some pieces may be
lost but the data itself can still be recovered.
This has many applications: Distributed storage
systems~\cite{DBLP:journals/tos/ShenCCLLHS25} allocate pieces of information
across multiple disks or storage nodes for increased resilience,
and in the neighboring field of reliable communication, network coding
protocols~\cite{DBLP:journals/ccr/FragouliBW06} split large files into blobs
that are subsequently transmitted individually for efficiency and can be
recombined at the receiver. %
In the realm of blockchains and cryptocurrencies, transaction data and system state have become so large that erasure codes are
used broadly to store and communicate such data efficiently.  For example,
some so-called data-availability networks use erasure codes to provide transaction data as a
service~\cite{DBLP:journals/cic/HallAndersenSW24}, and many recent
Byzantine-tolerant consensus protocols disseminate data with erasure
codes~\cite{DBLP:journals/csur/YangCWWLZ24}.

Many widely used erasure codes, such as maximum-distance separable (MDS) codes
(e.g., Reed--Solomon codes), divide their input data into $k$ pieces and encode
it into $n>k$ so-called \emph{fragments} of equal size such that any $k$ of the $n$ fragments
are sufficient to reconstruct the original data.  By allocating each fragment to
a distinct node in a set~\CP of cardinality~$n$ (e.g., disks, storage servers,
or consensus nodes) and by using cryptographic authentication (e.g., a Merkle
commitment or per-fragment hashes) to filter out invalid fragments, one can
recover the original data even if $n-k$ nodes fail.  This is more efficient than
replicating the full data $n-k+1$ times.

In many contexts, however, some nodes are more likely to fail than others or may
be more trustworthy than others.  In this case, allocating exactly one fragment
out of $n$ to each node~-- regardless of how likely it may fail or how
trustworthy it is~-- may not guarantee that $k$ fragments will be available.
This applies, in particular, to specific deployments of distributed cryptography
and secure multiparty computation, to voting power in blockchains that rely on
stake, and to the trust models of consensus protocols used in the XRP
Ledger~\cite{CohenSB25} and the Stellar
network~\cite{DBLP:conf/sosp/LokhavaLMHBGJMM19}.  In those scenarios, the
underlying assumption is typically captured by an access
structure~\cite{DBLP:conf/crypto/Leichter88,DBLP:journals/joc/HirtM00} or a
quorum system~\cite{DBLP:journals/dc/MalkhiR98}.

In this paper, we introduce \emph{monotone erasure codes} that respect arbitrary
trust assumptions on a set of nodes~\CP characterized by any \emph{access
structure}~$\CA \subseteq 2^\CP$; such codes assign fragments to nodes and
ensure that, for every \emph{access set}~$A$ contained in \CA, the fragments
assigned to the members of \(A\) are sufficient to reconstruct the data.  We
assume that all access sets are minimal, but every set of nodes in \CP that
extends an access set may also reconstruct the data; this makes the data
reconstruction capability monotone on~\CP.
In more detail, a monotone erasure code consists of a pair of algorithms called
\textsc{Encode} and \textsc{Decode}: Algorithm \textsc{Encode} receives a
\emph{file} as input and outputs a vector of $n$ fragments such that each node
in $\mathcal{P}$ receives one fragment.  In contrast to standard erasure codes, the
fragments may have different sizes.  The \textsc{Decode} algorithm, on input a
collection of fragments, outputs some information.  The code respects $\mathcal{A}$ in the
sense that \textsc{Decode} returns the original information whenever it receives
the fragments from an access set as input. %
Analogously to standard coding
theory, we focus on the important case of \emph{linear monotone erasure codes},
where the file and the fragments are vectors over a finite field $\BF_q$ and the
encoding involves only linear operations.

Important parameters for monotone erasure codes 
are the size of the input file, denoted $\kappa$, and the combined size of all fragments,
called~$\mu$.  Therefore, the efficiency of a code is measured by its
\emph{overhead} $\beta = \frac{\mu-\kappa}{\kappa}$, and this should be as
small as possible.  In other words, an efficient monotone erasure code adds as
little redundancy as possible.

In Section~\ref{sec:monotone_erasure_codes}, we first formalize the notion of a monotone erasure code and its overhead.
Then we provide a simple but inefficient method to
construct monotone erasure codes.
It works for any
access structure that is given by a monotone Boolean formula (MBF) with
threshold operators on \CP, i.e., from formulas in variables that represent the
nodes in~\CP.  Such an expression may not be a compact description of the access
structure; in this case, the resulting monotone erasure code can have a large
overhead.

As for other codes, monotone erasure codes that work over finite fields and where encoding is a linear operation are an important notion, as they have a rich mathematical structure.
We formalize them in Section~\ref{sec:linear} and show how they
generalize MDS erasure codes.
The most prominent MDS erasure codes are
Reed-Solomon codes~\cite{ReedS60,DBLP:journals/jacm/Rabin89};
we also call MDS codes threshold codes because their access structure consists
of all subsets of \CP with cardinality~$k$.

In Section~\ref{sec:matrix_construction}, we present an efficient construction of a linear monotone erasure code for any access structure. However, in general, this approach does not always achieve minimal storage overhead. Therefore, in Section~\ref{sec:construction_from_threshold} we focus on constructing \emph{optimal} linear monotone erasure codes. An algorithm to generate a linear monotone erasure code for an arbitrary access
structure from an MDS erasure code is provided. To optimize the overhead of the monotone erasure code, our algorithm involves solving a rational linear programming problem whose complexity depends on the number of access sets.

An important class of non-threshold access structures used in practice are access structures stemming from a hierarchy of domains, and we examine
them in Section~\ref{sec:partitioned_access}.  Such access structures form
a hierarchy of entities, where each entity at one level is comprised of
multiple entities at a lower level.  The Stellar network~\cite{DBLP:conf/sosp/LokhavaLMHBGJMM19} is a real-world system using this kind of access structures.
Formally, these correspond to MBFs where each variable appears at most once.  
We give an efficient algorithm to construct an optimal monotone erasure code for such a \emph{partitioned access structure}, and characterize its overhead.

In Section~\ref{sec:gavid}, we turn to distributed protocols.
Erasure codes are found in many Byzantine consensus and broadcast
algorithms for blockchains today~\cite{DBLP:conf/podc/CamenischDHPS022,
DBLP:conf/wdag/Shoup24, DBLP:conf/eurocrypt/LocherS25, KniepSW25, DBLP:journals/csur/YangCWWLZ24}
because they allow for efficient information dispersal, often using protocols in
the vein of the protocol of Cachin and Tessaro~\cite{DBLP:conf/srds/CachinT05}
for \emph{asynchronous verifiable information dispersal (AVID)}.  We show how to
use monotone erasure codes to generalize this classic protocol to obtain a
communication-efficient \emph{generalized AVID} protocol that works not only for Byzantine quorum systems that tolerate $f$ failures out of $n$ nodes but also for general Byzantine quorum systems. The key ingredients that allow the generalization are, of course, monotone erasure codes, but also the notion of \emph{kernels} (generalizing sets of \(f+1\) nodes) and our novel notion of \emph{reliable sets}.
From generalized AVID, or GAVID, one can also obtain a communication-efficient
Byzantine reliable broadcast protocol for an arbitrary Byzantine quorum
system.  This paves the way to integrating monotone erasure codes into other
protocols like consensus or into distributed storage systems.

Finally, erasure codes for reducing communication complexity in distributed
computing have become a prominent topic recently. We review some of the
literature and compare our notions to those from related work in Section~\ref{sec:related_work}.

\paragraph{Contributions.} The four principal contributions of this work
are:
\begin{itemize}
\item To propose monotone erasure codes that generalize classical erasure
  codes, which address only threshold access structures;
\item To provide two methods to build a linear monotone erasure code for any
  monotone access structure: one that is computationally efficient and one that minimizes the overhead of the code;
\item To analyze the special case of a partitioned access structure,
  for which we build a linear monotone erasure code with optimal overhead; and
\item To formulate a general asynchronous verifiable information dispersal
  protocol (GAVID) for building communication-efficient Byzantine-tolerant
  reliable broadcast with non-threshold quorum systems.
\end{itemize}

\section{Access Structures}
\label{sec:access_struct}

Let $\mathcal{P}$ be a set of $n$ nodes $\{p_1, \ldots, p_n\}$. For each $n \in \mathbb{N}$, we denote the set $\{1,\ldots, n\}$ by $[n]$. 
In many applications, it is necessary to specify which groups of nodes are able to perform certain actions.
For example, in a distributed storage system, we want to ensure resilience against data loss at some nodes. In other words, even if certain nodes lose their stored data, the set of remaining nodes should still be able to recover the original information.
The collection of all such sets is called an \emph{access structure}.

\begin{definition}[Access structure]\label{def:access.structure}
    An \emph{access structure} $\mathcal{A}$ on $\CP$ is a collection of subsets of~$\mathcal{P}$ such that no set is contained in another. Each set $A \in \mathcal{A}$ is called an \emph{access set}.
\end{definition}

Note that in the literature, the structure described in Definition~\ref{def:access.structure} is usually called a \emph{minimal} access structure. However, in our context, if an access set is able to reconstruct the stored data, then all supersets of this set can also reconstruct it. This makes the ability to reconstruct data \emph{monotone} in \CP. Thus, it suffices to consider only the minimal collection of access sets and show that these sets can reconstruct the data.

It is possible to represent any monotone access structure using a \textit{monotone Boolean formula} (MBF). An MBF is a formula composed of AND, OR, and threshold operators, along with atoms. %
In our case, each atom corresponds to a node $p_i$. The threshold operator $\Theta_k^m(q_1, \ldots, q_m)$ evaluates to 1 whenever $k$ out of the $m$ functions $q_1, \ldots, q_m$ evaluate to 1, where each $q_i$ can be either an atom or an operator. Notice that $\Theta_m^m$ corresponds to the AND operator and $\Theta_1^m$ corresponds to the OR operator evaluated on $m$ inputs. %

An MBF $\Gamma$ describes a monotone function $\varphi :2^\mathcal{P} \to \{0,1\}$ in the following way: 
\begin{itemize}
    \item when $\Gamma$ consists only of an atom, then the value of $\Gamma$ on input $S \subseteq \CP$ is $1$ if and only if $\Gamma \in S$;
    \item when $\Gamma$ is the threshold operator 
$\Theta^m_k(q_1, \ldots, q_m)$, then $\Gamma(S) = 1$ if at least $k$ of the $q_1, \ldots, q_m$ are recursively evaluated to $1$ on input $S$. 
\end{itemize}

Note that a monotone access structure $\mathcal{A}$ can be described using a monotone function $\alpha :2^\mathcal{P} \to \{0,1\}$ that returns 1 for every access set and 0 otherwise. We can therefore represent the access structure using the MBF that describes $\alpha$. 
We refer to the rooted tree representing such an MBF as the \emph{access tree} of $\mathcal{A}$.
In more detail, an access tree is a rooted, labeled tree in which each internal vertex \(v\) has a \emph{label} that is an integer between 1 and the number of children of \(v\) (inclusive), and each leaf vertex is labeled with a node identifier.

An important subclass of access structures is what we call \emph{partitioned access structures}.
Those are access structures that can be described by an MBF where each atom appears at most once. Note that in the access tree of a partitioned access structures each leaf is labeled with a unique node identifier.

\section{Monotone Erasure Codes}
\label{sec:monotone_erasure_codes}

\subsection{Definition}
\label{ssec:definition}

Consider an access structure $\mathcal{A}$ on $\mathcal{P}$. Suppose we want to store a file~$f$, also called an information vector, among the nodes in~$\mathcal{P}$, with the guarantee that every access set $A$ is able to reconstruct~$f$. %
Additionally, we seek an efficient storage scheme, meaning that the total amount of information stored across all nodes is as small as possible.
To address this need, we propose a new kind of erasure code, called \textit{monotone erasure code}. Let $\mathcal{F}$ be a finite set and let $\mathcal{G}=\mathcal{G}_1 \times \ldots \times \mathcal{G}_n$, where each $\mathcal{G}_i$ is a finite set. 

\begin{definition}[Monotone erasure code]\label{Monotone erasure code}
     Let $\mathcal{A}$ be an access structure on $\mathcal{P}$ and let $f$ be a file from a set $\mathcal{F}$. A \textit{monotone erasure code $\mathcal{C}$ for $\mathcal{A}$} is a scheme that distributes $f$ among the nodes in $\mathcal{P}$ and assigns to each $p_i \in \mathcal{P}$ a piece of information $g_i \in \mathcal{G}_i$ called \textit{fragment} of $p_i$.
    Such a code consists of two functions:
\begin{description}
    \item \textsc{Encode}$(f)$: given a file $f$, 
    returns a fragment vector $g=(g_1, \ldots, g_n) \in \mathcal{G}$; %
    \item \textsc{Decode}$(u)$: given %
    a vector $u=(u_i)_{i \in [n]} \in \mathcal{G}$, %
    returns either a file $f' \in \mathcal{F}$ or $\bot$.
\end{description}
A monotone erasure code for $\mathcal{A}$ satisfies the following property:
\begin{description}
    \item \textbf{Completeness:} for all $A\in\mathcal{A}$, the fragments $g_A$ of the nodes in $A$ generated by \textsc{Encode}($f$) are enough to reconstruct $f$, i.e., $g_A$ satisfies $$\textsc{Decode}(g_A)=f.$$ 
\end{description}
\end{definition}
For any $B \subseteq \CP$, we denote by $g_B$ the vector whose $i$-th entry equals $g_i$ if $p_i \in B$, and $\bot$ otherwise.
Moreover, we assume w.l.o.g. that $\bot \in \mathcal{G}_i$, for all $i \in [n]$, and allow that $\mathcal{G}_i = \{\bot\}$, meaning that the fragment of $p_i$ is the constant $\bot$. We introduce this constant because, for certain structures, it may be convenient to assign no actual data to some nodes, in which case their fragment is $\bot$.

Let $g=\text{\textsc{Encode}}(f)$ for some file $f$. Then, we call a set $B\subseteq \mathcal{P}$ \textit{sufficient} if \textsc{Decode}$(g_B)=f$. Otherwise, if \textsc{Decode}$(g_B) \neq f$, i.e., \textsc{Decode}$(g_B)$ outputs a vector different from $f$ or \textsc{Decode}$(g_B)= \bot$, then $B$ is called \textit{insufficient}. %
In particular, this means that each access set $A$ is sufficient, namely, it is able to reconstruct the original file $f$ using the fragments $g_A$. However, we do not require all sufficient sets to be in $\mathcal{A}$. As a consequence, it may happen that some sets are able to reconstruct $f$, even though they are not in $\mathcal{A}$. This is an important difference from secret sharing, where sets outside $\mathcal{A}$ are only able to recover negligible or no information.

A monotone erasure code does not provide integrity, in the sense that a group of nodes $B$ does not know if the retrieved file is the original one. In other words, $B$ does not know if $f'=f$, where $\textsc{Decode}(g_B)=f'$ and $f' \neq \bot$. Nevertheless, this issue can be fixed by employing a verification method, such as a Merkle tree.

The efficiency of a monotone erasure code used to store a file $f\in \{0,1\}^\kappa$ can be measured through a metric called \textit{overhead}. It measures the amount of redundancy added per information symbol during the encoding phase.

\begin{definition}[Overhead]
    Let $\mathcal{C}$ be a monotone erasure code that takes as input files of $\kappa$ bits and produces fragments of $\mu$ total bits. The \textit{overhead} $\beta$ of $\mathcal{C}$ is $\beta=\frac{\mu-\kappa}{\kappa}$.
\end{definition}

\subsection{A basic construction of a monotone erasure code}
\label{sec:basic}

In this section, we present an algorithm that, given an access structure \(\mathcal{A}\) described as a MBF, builds a monotone erasure code that is complete for $\mathcal{A}$.
For illustration purposes, we assume here that each operator in the MBF is either AND or OR, but note that threshold operators can always be converted into a composition of AND and OR. In other words, the formula can be expressed as  
\[
F(F^{(1)}_1(F^{(2)}_1(\ldots), \ldots, F^{(2)}_{\ell_1}(\ldots)), \ldots, F^{(1)}_\ell(\ldots))
\]

where $F^{(i)}_j$ is an AND or OR operator at depth $i$. 

The goal is to assign enough information to each node to achieve completeness. In particular, let $f \in \{0,1\}^\kappa$ be the file we want to store. 
Then, the construction works as follows. Consider the access tree $T$ representing $\mathcal{A}$ and let $\mathcal{V}$ be the set of all vertices in $T$. The first step of the algorithm %
is to build a function $\eta: \mathcal{V} \xrightarrow[]{} \{0,1\}^*$ that assigns a substring of $f$ to each vertex.%

We start by considering the root of the tree, which corresponds to the operator $F$, and we set $\eta(F)=f$. Then, if $F$ is an AND with $\ell$ arguments, we divide $f$ in $\ell$ different chunks $f_1, \ldots, f_\ell$ with length $\frac{\kappa}{\ell}$, assuming w.l.o.g. that such fractions are integers.
We then set $\eta(v_i)=f_i$ for each child $v_i$ of the root.
Otherwise, if $F$ is an OR, we set $\eta(v_i)=f$ for each child $v_i$. 
Next, we move to the first child of $F$ and we 
iterate the procedure until we have visited each leaf of the tree and defined $\eta$ for it. 

Notice that each node $p_i \in \mathcal{P}$ might correspond %
to multiple leaves $v_{i_1}, \ldots, v_{i_{w_i}}$. The fragment of $p_i$ is then defined as the list of the chunks $\eta(v_{i_1}), \ldots, \eta(v_{i_{w_i}})$. %
This ensures that each access set has enough information to reconstruct $f$.

\begin{example}
    Let $\mathcal{P} = \{a,b,c,d,e \}$ be the set of nodes and \[\mathcal{A}=\{\{a,c,d,e\}, \{b,c,d,e\}, \{a,b,c\}, \{a,b,d\}, \{a,b,e\}\}\] be the access structure. Then, the corresponding MBF is 
\[\Theta_1^2(\Theta_2^2(\Theta_1^2(a,b), \Theta_3^3(c,d,e)) , \Theta_2^2(\Theta_2^2(a,b), \Theta_1^3(c,d,e))).\]
Suppose we want to store the file $f=(0,1,1,1,0,0,1,0,1,1,0,1) \in \{0,1\}^{12}$.
\begin{figure}[h!]
\centering
\begin{tikzpicture}[grow=down, level distance=2cm]
\node[draw, circle, label=left:{$f$}]{$\vee$} %
  child [sibling distance=5cm]{node[draw, circle, label=left:{$f$}]{$\wedge$}
    child [sibling distance=2.5cm]{node[draw, circle, label=left:{$f_1$}]{$\vee$}
        child [sibling distance=1cm]{node[draw, circle, label=below:{$f_1$}]{a}}
        child [sibling distance=1cm]{node[draw, circle, label=below:{$f_1$}]{b}}}
    child [sibling distance=2.5cm]{node[draw, circle, label=left:{$f_2$}]{$\wedge$}
        child [sibling distance=1cm]{node[draw, circle, label=below:{$f'_1$}]{c}}
        child [sibling distance=1cm]{node[draw, circle, label=below:{$f'_2$}]{d}}
        child [sibling distance=1cm]{node[draw, circle, label=below:{$f'_3$}]{e}}}
  }
  child [sibling distance=5cm]{node[draw, circle, label=left:{$f$}]{$\wedge$}
    child [sibling distance=2.5cm]{node[draw, circle, label=left:{$f_1$}]{$\wedge$}
        child [sibling distance=1cm]{node[draw, circle, label=below:{$f''_1$}]{a}}
        child [sibling distance=1cm]{node[draw, circle, label=below:{$f''_2$}]{b}}}
    child [sibling distance=2.5cm]{node[draw, circle, label=left:{$f_2$}]{$\vee$}
        child [sibling distance=1cm]{node[draw, circle, label=below:{$f_2$}]{c}}
        child [sibling distance=1cm]{node[draw, circle, label=below:{$f_2$}]{d}}
        child [sibling distance=1cm]{node[draw, circle, label=below:{$f_2$}]{e}}}
  };
\end{tikzpicture}
\caption{Access tree representing $\mathcal{A}$ and $\eta$.}
\label{Access tree 1}
\end{figure}
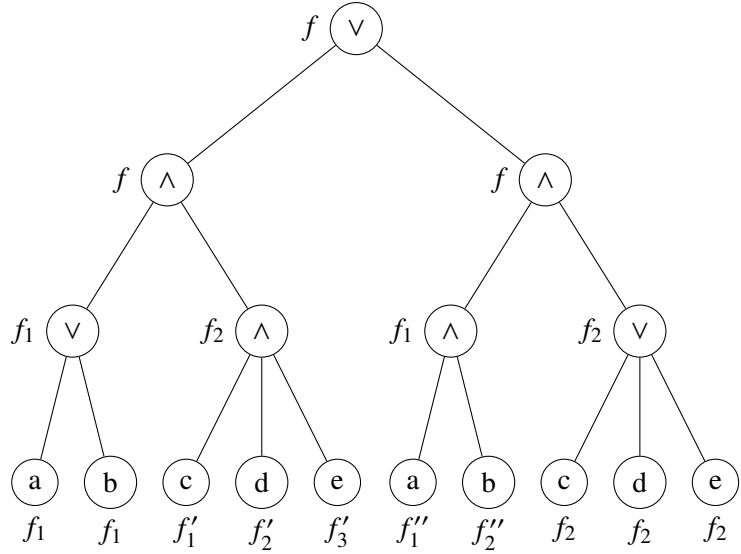 
Figure \ref{Access tree 1} shows the access tree that represents $\mathcal{A}$, along with the evaluation of $\eta$ on the nodes, where 
\begin{multicols}{3}
    \begin{itemize}
    \item $f_1=(0,1,1,1,0,0)$
    \item $f_2=(1,0,1,1,0,1)$
    \item $f'_1=(1,0)$
    \item $f'_2=(1,1)$
    \item $f'_3=(0,1)$
    \item $f''_1=(0,1,1)$
    \item $f''_2=(1,0,0)$
\end{itemize}
\end{multicols}
 \hfill $\triangle$
\end{example}

\section{Linear Monotone Erasure Codes}
\label{sec:linear}

We now consider monotone erasure codes for data represented by symbols that are elements of a finite field~$\mathbb{F}_q$ .
This means the file is a vector $f \in \mathbb{F}_q^k$ of length~$k$.
Moreover, we take the fragments $(g_1, \ldots, g_n)$ from the set
$\mathcal{G}= \mathcal{G}_1 \times \cdots \times \mathcal{G}_n$, where
$\mathcal{G}_i= \mathbb{F}_q^{m_i} \cup \{\bot\}$ or
$\mathcal{G}_i= \{\bot\}$.  Then $m_i\in \mathbb{N}$ denotes the size of
$p_i$'s fragment, and we set $m_i=0$ whenever $\mathcal{G}_i=\{\bot\}$.  The
total size of the stored information is $m=\sum_{i=1}^n m_i$ symbols.

\begin{definition}[Linear monotone erasure code]
  An $[m,k]$-\emph{linear monotone erasure code $\mathcal{C}$ for
    $\mathcal{A}$} is a monotone erasure code for an access structure \CA and
  files of length~$k$ that produces fragments with combined size~$m$, where
    the encoding and decoding algorithms are based on a full-rank matrix $G\in \mathbb{F}_q^{k \times m}$ and a function $\phi:\{1,\ldots,m\}\rightarrow \mathcal{P}$ that assigns each column of $G$ to a node in $\mathcal{P}$. Let $G_{p_i}$ be the matrix consisting of the columns of $G$ assigned to $p_i$ by $\phi$, i.e., 
    \[
    G_{p_i}= ( G_{j_1} | \cdots |G_{j_{m_i}} )
    \]
    where $\phi(j_1) = \ldots = \phi(j_{m_i})=p_i$. Then, the algorithms are defined as follows:  
     
      \begin{description}
        \item \textsc{Encode}($f$): Given $f\in \mathbb{F}_q^k$, the algorithm outputs the fragments $(g_1,\ldots,g_n)$, where $g_i \in  \mathcal{G}_i$ such that 
        $$g_i=f \cdot G_{p_i}$$
        is the fragment of $p_i\in \mathcal{P}$. If there is no $j \in [m]$ such that $\phi(j)=  p_i$, i.e., $m_i=0$, then $g_i = \bot$.

        \item \textsc{Decode}($u$): Given a vector $u \in \mathcal{G}$, the algorithm returns either $f' \in \mathbb{F}^{k}_q$ or $\bot$.
    \end{description}
    In relation to the classical erasure codes, the matrix $G$ is also called
    the \emph{generator matrix} of~\CC, $k$ is the \emph{dimension} of~\CC,
    and $m$ is the \emph{length} of~\CC.
\end{definition}

According to the earlier definitions, an $[m,k]$-linear monotone erasure code is complete if for all $A\in\mathcal{A}$, the fragments $g_A$ of the nodes in $A$ contain enough information to reconstruct $f$, i.e., \textsc{Decode}$(g_A)=f$. Moreover, let $G_B$ be the matrix given by the columns of $G$ assigned to nodes in a set $B \subseteq \CP$. Then, $B$ is sufficient if $f\cdot G_B$ uniquely determines the information vector~$f$. Otherwise, $B$ is insufficient. 

By the completeness property of a linear monotone erasure code $\mathcal{C}$, each set $A\in \mathcal{A}$ is sufficient. Moreover, by the following theorem, it follows that the labeling function must assign to the nodes in a set $A\in \mathcal{A}$ at least $k$ linearly independent columns of $G$ to make $\mathcal{C}$ complete.
\begin{theorem} \label{Theorem: linear monotone erasure code sufficient property}
    Let $\mathcal{C}$ be an $[m,k]$-linear monotone erasure code over $\mathbb{F}_q$ for $\mathcal{A}$. A set of nodes $B\subseteq \mathcal{P}$ is sufficient if and only if the matrix $G_B$
    has full rank, i.e., $\op{rank}(G_B)=k$.
\end{theorem}
\begin{proof}
    Assume $G_B \in \mathbb{F}_q^{k\times l}$  where $l\geq k$. Since $\op{rank}(G)=k$ it follows that $\op{rank}(G_B)=k$. Let $\hat{G}_B$ be a matrix containing $k$ columns of $G_B$ that are linearly independent. Then $\hat{G}_B\in \mathbb{F}_q^{k\times k}$ is invertible and thus there exists $U\in \mathbb{F}_q^{k\times k}$ such that $\hat{G}_B\cdot U = I_k$, where $I_k$ is the $k\times k$ identity matrix. Then we have $$(f\cdot \hat{G}_B)\cdot U = f\cdot I_k=f.$$  If $G_B\in \mathbb{F}_q^{k\times l}$, with $l<k$, then $\op{rank}(G_B)=l$. Therefore, there exist $x_1,\ldots,x_k\in \mathbb{F}_q$ not all equal to $0$ such that $x_1g_1+x_2g_2+\ldots +x_kg_k=0$ where $g_1,g_2,\ldots,g_k$ are the rows of $G_B$.  Let $x=(x_1,\ldots,x_k)\in \mathbb{F}_q^k$. Then, $x+f=(x_1+f_1,\ldots,x_k+f_k)\in \mathbb{F}_q^k$ and $(x+f) \cdot G_B=(x_1+f_1)g_1+\ldots+(x_k+f_k)g_k=x_1g_1+\ldots+x_kg_k+ f_1g_1+\ldots+f_kg_k = f_1g_1+\ldots+f_kg_k=f \cdot G_B$. Therefore, $f \cdot G_B$ does not uniquely determine $f$.
\end{proof}

To store an information vector $f\in \mathbb{F}_q^k$ among the nodes, we are interested in the efficiency of a linear monotone erasure code $\mathcal{C}$, in the sense of how much more data must be stored compared to $f$. As previously mentioned, this can be measured using the overhead $\beta$ of $\mathcal{C}$. In this case, it holds that $\kappa=k \cdot \lceil log_2(q) \rceil$ and $\mu= m \cdot \lceil log_2(q) \rceil$, thus $\beta= \frac{m-k}{k}$.

If $k$ is close to $m$, then the overhead $\beta$ is small. 
This is equivalent to adding only little redundancy, since $m-k$ is small.
Therefore, to efficiently store a file $f$, the goal is to find the optimal parameters of a linear monotone erasure code $\mathcal{C}$, i.e., $k$ and $m$ of $\mathcal{C}$ that minimize the overhead. In this way, for a given access structure $\mathcal{A}$ and optimal parameters $k$ and $m$, we represent $f$ as an information vector of length $k$ over $\mathbb{F}_q$.  
 
For some specific access structures, the optimal parameters $k$ and $m$ can be found easily. In particular, if $\mathcal{A}$ is an access structure such that each set $A\in \mathcal{A}$ consists of exactly $w$ nodes, for some $w \in [n]$, then one can find the optimal parameters as follows, which leads us to the definition of a \textit{linear threshold erasure code}. This is well-known in the literature on \textit{information dispersal} \cite{DBLP:journals/jacm/Rabin89, DBLP:conf/srds/CachinT05} where a data bit string is encoded into $n$ pieces and stored among $n$ nodes such that each node stores exactly one piece. 

\begin{definition}[Linear threshold erasure code]
    Let $\mathcal{P}$ be a set of $n$ nodes and let $\mathcal{A}$ be an access structure on $\mathcal{P}$ such that each $A$ in $\mathcal{A}$ consists of exactly $w$ nodes, i.e., $\mathcal{A}=\{ A \text{ } \big| \text{ } |A|=w\}$. 
    An $[m,k]$-\textit{linear threshold erasure code $\mathcal{C}$ 
    for $\mathcal{A}$} is an $[m,k]$-linear monotone erasure code over $\mathbb{F}_q$ for $\mathcal{A}$, with parameters $k=w$ and $m=n$, and the labeling function 
        \begin{align*}
        \phi(i)= p_i \text{ for } i\in [n].
        \end{align*}
\end{definition}

In a linear threshold erasure code, the fragment $g_i$ of each node $p_i$ is made of exactly one field element, which we call \textit{threshold fragment}.
Moreover, the nodes of each set $A\in \mathcal{A}$ have together $k$ different columns of $G$. This means that the overhead is $\frac{n-k}{k}$, which is the best we can achieve for this specific access structure. The reason is that the nodes do not store more information than is required to reconstruct $f$.

We now consider a particular and well-known  generator matrix $G$ for an $[m,k]$-linear threshold erasure code $\mathcal{C}$. Let $q\geq m \geq k$ be integers and let $\alpha = (\alpha_1,\ldots,\alpha_m)\in \mathbb{F}_q^m$ where $\alpha_i \neq \alpha_j$ for all $i\neq j \in [m]$. Then, 
    \begin{align*}
        G=\begin{pmatrix}
            1 & 1 & \ldots & 1\\
            \alpha_1 & \alpha_2 &  \ldots & \alpha_m\\
            \vdots & & &\vdots \\
            \alpha_1^{k-1} & \alpha_2^{k-1} & \ldots & \alpha_m^{k-1}
        \end{pmatrix}\in \mathbb{F}_q^{k\times m}.
    \end{align*}

The matrix $G$ is called \textit{Vandermonde matrix} and generates a \textit{maximum distance separable} (MDS) code, called \textit{Reed-Solomon code}.

\begin{definition}[MDS code]
    A \textit{linear code} $\mathcal{C}$ of length $m$ and dimension $k$ is called  \textit{maximum distance separable} (MDS) if $d=m-k+1$ where $d=\min\{d_H(x,y) \text{ }| \text{ } x, y\in \mathcal{C}, x\neq y\}$ with    
    $d_H(x,y)=|\{i\in [n] \text{ } | \text{ } x_i
    \neq y_i\}|$ is called \textit{minimum distance} of $\mathcal{C}$.
\end{definition}

We refer to an MDS code with length $m$ and dimension $k$ as an $[m,k]$-MDS code.
Notice that any $[m,k]$-MDS code can be viewed as an $[m,k]$-linear threshold erasure code. Indeed, even though the labeling function is not part of the standard definition of an MDS code, it can be naturally added. In particular, one can define a set of $m$ virtual nodes $\mathcal{V}=\{V_1, \ldots, V_m\}$ and set $\phi(i)=V_i$, for all $i \in [m]$. We therefore consider MDS codes to be linear threshold erasure codes.

In particular, any generator matrix of an MDS code can be chosen for $G$. This offers an optimal trade-off between fault-tolerance and storage overhead, as it allows data to be reconstructed from any $k$ threshold fragments, derived from $k$ different columns of $G$. 

Indeed, for an MDS code it holds that any $k$ columns of its generator matrix $G$ are linearly independent. Therefore, if the nodes of a set $B\subseteq \CP$ have together at least $k$ different columns of $G$, then they are able to reconstruct $f$ by Theorem~\ref{Theorem: linear monotone erasure code sufficient property}.

\section{Efficient Construction of Linear Monotone Erasure Codes}
\label{sec:matrix_construction}

In this section, we present an efficient algorithm that constructs a linear monotone erasure code $\CC$ for any given access structure $\CA$ on $\mathcal{P}$. The encoding matrix of $\CC$ has a block-wise MDS property that ensures completeness of the code. In the general case, this construction does not deliver a linear monotone erasure code with minimal overhead.

We first present the intuition behind the code through its encoding algorithm. 
Given a file~$f$ and an access tree $T$ with a root node labeled $t$ that has $r$ children, the
algorithm proceeds as follows.
First, it encodes $f$ using a linear monotone erasure code that generates $r$ fragments and such that any $t$ of them can reconstruct~$f$. Then, it encodes each of the $r$ resulting fragments again. Specifically, for each child~$v_i$ of the root with label~$t_i$ and $r_i$ children, the fragment~$g_i$ assigned to $v_i$ is encoded using a linear monotone erasure code that generates $r_i$ fragments and such that any $t_i$ of them are enough to reconstruct $g_i$. The encoding function proceeds in this way until it reaches the leaves. 
This method is inspired by secret-sharing constructions in the literature (e.g., \cite{cryptoeprint:2004/282}).

To build such a code, we use a recursive algorithm.
We now illustrate the functioning of the algorithm through an example before describing it in detail. Consider the tree $T$ in Figure~\ref{fig:access_tree_efficient_construction}, representing the access structure $$\mathcal{A}=\Theta_2^3(\Theta_3^3(p_1,p_2,p_4),\Theta_2^3(p_3,p_4,p_5),\Theta_1^3(p_6,p_7,p_8)).$$

    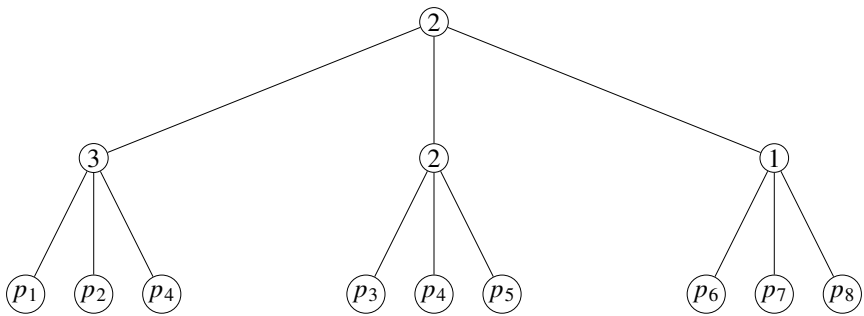
\begin{figure}[h]
    \centering
    \scalebox{0.9}{\begin{tikzpicture}[grow=down, level distance=2cm]
    \tikzset{
    circ/.style = {
        draw,
        circle,
        inner sep = 1pt}}
        
\node[circ]{$2$} %
  child [sibling distance=5cm]{node[circ]{$3$}
                   child [sibling distance=1cm]{node[circ]{$p_{1}$}}
            child [sibling distance=1cm]{node[circ]{$p_{2}$}}
           child [sibling distance=1cm]{node[circ]{$p_{4}$}} }
  child [sibling distance=5cm]{node[circ]{$2$}
            child [sibling distance=1cm]{node[circ]{$p_{3}$}}
            child [sibling distance=1cm]{node[circ]{$p_{4}$}}
            child [sibling distance=1cm]{node[circ]{$p_{5}$}}}
            child [sibling distance=5cm]{node[circ]{$1$}
            child [sibling distance=1cm]{node[circ]{$p_{6}$}}
            child [sibling distance=1cm]{node[circ]{$p_{7}$}}
            child [sibling distance=1cm]{node[circ]{$p_{8}$}}}
 ;
\end{tikzpicture}}
\caption{Access tree $T$ of $\mathcal{A}$}\label{fig:access_tree_efficient_construction}

\end{figure}

Let $T_i$ be depth-one subtrees of $T$ for $i\in [3]$ with root labeled $t_i$ with $r_i$ children. We aim to construct an encoding matrix $M\in \mathbb{F}_q^{k\times m}$ for a sufficient large $q$ and a labeling function $\phi$ for $\mathcal{A}$ such that the nodes of each access set yield at least $k$ linearly independent columns of $M$. 
To achieve this, the matrix construction guarantees the following. Each node in $T_i$ yields $k/(2t_i)$ linearly independent columns, so that each subtree $T_i$ collectively contributes a set of $k/2$ linearly independent columns. Moreover, the union of columns of any two such sets remains linearly independent. We pick the order of the field $q$ as a prime power greater than or equal to the following value: 
$$ \max\bigl \{2,\; \max_{v_i, \; t_i \neq 1} \{r_i \}\bigr\}= 3$$
where each $v_i$ is a vertex in $T$ labeled $t_i$ with $r_i$ children. 

To construct a matrix $M$ for $\mathcal{A}$ that satisfies the property above, we first associate to each $T_i$ a $t_i\times r_i$ Vandermonde matrix $M_i$ over $\mathbb{F}_q$. By construction, any $t_i$ columns of $M_i$ are linearly independent. Thus, assigning one column of $M_i$ to each node in $T_i$ ensures that any $t_i$ nodes contribute sufficiently many linearly independent columns. 

\begin{center}
	\begin{tikzpicture}
		\node (M1) at (0,0) {$M_1 = \begin{pmatrix}
		1 & 1 & 1\\
		0 & 1 & 2\\
		0 & 1 & 1
	\end{pmatrix}$};
	
	\draw[thick, ->] (-0.05,-0.75) -- (-0.05,-1.25) node[below] {$p_1$};
	\draw[thick, ->] (0.45,-0.75) -- (0.45,-1.25) node[below] {$p_2$};
	\draw[thick, ->] (0.95,-0.75) -- (0.95,-1.25) node[below] {$p_4$};
	
	\node (M2) at (4,0) {$M_2 = \begin{pmatrix}
		1 & 1 & 1\\
		0 & 1 & 2
	\end{pmatrix}$};
	
	\draw[thick, ->] (3.95,-0.5) -- (3.95,-1.25) node[below] {$p_3$};
	\draw[thick, ->] (4.45,-0.5) -- (4.45,-1.25) node[below] {$p_4$};
	\draw[thick, ->] (4.95,-0.5) -- (4.95,-1.25) node[below] {$p_5$};
	
	\node (M3) at (8,0) {$M_3 = \begin{pmatrix}
		1 & 1 & 1
	\end{pmatrix}$};
	
	\draw[thick, ->] (7.95,-0.2) -- (7.95,-1.25) node[below] {$p_6$};
	\draw[thick, ->] (8.45,-0.2) -- (8.45,-1.25) node[below] {$p_7$};
	\draw[thick, ->] (8.95,-0.2) -- (8.95,-1.25) node[below] {$p_8$};	
	\end{tikzpicture}
\end{center}

To combine $M_1, M_2$ and $M_3$ so that the resulting matrix $M$ satisfies the desired property, we have to lift the matrices $M_i$ to a common number of rows (which we set as the least common multiple). We achieve this by computing the Kronecker product between the matrix $M_i$ and an appropriate identity matrix, resulting into the following matrices $R_i$ with corresponding labeling of the columns:

\begin{center}
	\begin{tikzpicture}
		\node (R1) at (0,0) {$R_1 = M_1 \otimes I_2 = \begin{pmatrix}
		I_2 & I_2 & I_2\\
		\textbf{0} & I_2 & 2I_2\\
		\textbf{0} & I_2 & I_2
	\end{pmatrix}$,};
	
	\draw[thick, ->] (0.6,-0.75) -- (0.6,-1.25) node[below] {$p_1$};
	\draw[thick, ->] (1.2,-0.75) -- (1.2,-1.25) node[below] {$p_2$};
	\draw[thick, ->] (1.8,-0.75) -- (1.8,-1.25) node[below] {$p_4$};
	
	\node (R2) at (5,0) {$R_2 = M_2 \otimes I_3 = \begin{pmatrix}
		I_3 & I_3 & I_3\\
		\textbf{0} & I_3 & 2I_3
	\end{pmatrix}$,};
	
	\draw[thick, ->] (5.6,-0.5) -- (5.6,-1.25) node[below] {$p_3$};
	\draw[thick, ->] (6.2,-0.5) -- (6.2,-1.25) node[below] {$p_4$};
	\draw[thick, ->] (6.8,-0.5) -- (6.8,-1.25) node[below] {$p_5$};
	
	\node (R3) at (10,0) {$R_3 = M_3 \otimes I_6 = \begin{pmatrix}
		I_6 & I_6 & I_6
	\end{pmatrix}$};
	
	\draw[thick, ->] (10.6,-0.2) -- (10.6,-1.25) node[below] {$p_6$};
	\draw[thick, ->] (11.25,-0.2) -- (11.25,-1.25) node[below] {$p_7$};
	\draw[thick, ->] (11.9,-0.2) -- (11.9,-1.25) node[below] {$p_8$};	
	\end{tikzpicture}
\end{center}
By construction of $R_i$, any $t_i$ nodes contribute a set of linearly independent columns of equal size, as they are derived from linearly independent columns of $M_i$ via the Kronecker product.

To construct the final matrix $M$, we embed the matrices $R_i$ for $i\in[3]$ into the auxiliary matrix $\tilde{M}$, a $2\times 3$ Vandermonde matrix over $\mathbb{F}_q$:

\begin{center}
	\begin{tikzpicture}
		\node (MT) at (0,0) {$\Tilde{M} = \begin{pmatrix}
		1 & 1 & 1\\
		0 & 1 & 2
	\end{pmatrix}$};
	
	\draw[thick, ->] (-0.05,-0.5) -- (-0.05,-1.25) node[below] {$T_1$};
	\draw[thick, ->] (0.45,-0.5) -- (0.45,-1.25) node[below] {$T_2$};
	\draw[thick, ->] (0.95,-0.5) -- (0.95,-1.25) node[below] {$T_3$};	
	\end{tikzpicture}
\end{center}
Each column of $\tilde{M}$ corresponds to a subtree $T_i$ of $T$. We finally obtain $M$ by substituting each entry of $\tilde{M}$ by its Kronecker product with the corresponding matrix $R_i$. This results in the following block-structured matrix

\begin{center}
	\begin{tikzpicture}
		\node (M) at (0,0) {$M = \begin{pmatrix}
		R_1 & R_2 & R_3\\
		\textbf{0} & R_2 & 2R_3
	\end{pmatrix}$};
	\end{tikzpicture}
\end{center}

Consequently, each subtree $T_i$ contributes a set of $k/2$ columns, and the sets of columns of any two subtrees $T_i$ and $T_j$ are jointly linearly independent, ensuring that $M$ possesses the desired property.

\subsection{Algorithm}~\label{alg:matrix_construction} 
Given the access tree $T$ of $\CA$,
we set the field size $q$ as the smallest prime power such that $$q \geq \max\bigl\{ 2,\; \max_{v_i, \; t_i \neq 1} \{r_i  \}\bigr\} $$ where
$v_i$ is a vertex in $T$ labeled $t_i$ with $r_i$ children. 
Then, we use a recursive algorithm that takes the tuple \( (T, q) \) as input and returns the tuple \((M_T,\nu_T,k_T, L_T)\) defining $\CC$, where:

\begin{itemize}
	\item $M_T$ is a $k_T\times \nu_T$-matrix over $\mathbb{F}_{q}$ 
	\item $L_T$ is a vector defining the labeling function $\phi_T$ of $\CC$, i.e.,
	\begin{align*}
		\phi_T:[\nu_T]&\rightarrow \mathcal{P}\\
		i&\mapsto L_T(i)
	\end{align*}
\end{itemize}
\noindent
The algorithm works as follows. Note that all operations are done in $\BF_q$. 
\begin{enumerate}
	\item If $T$ has depth zero, i.e., it consists of a single vertex $p_i$, then let $(\nu_T,k_T)=(1,1)$, set 
	 \begin{align*}
				M_T=\left(\begin{array}{c}
			1 
			\end{array}
			\right)_{1\times 1} \qquad \text{ and } \qquad  L_T=(p_i)
			\end{align*}
			 and return the tuple $(M_T,\nu_T,k_T,L_T)$.
	\item If $T$ has depth one or more and consists of a root labeled $t$ with $r$ subtrees $T_1,\ldots,T_r$:
	\begin{enumerate}
		\item Apply the algorithm to each subtree $T_a$ for $a\in [r]$ to obtain the tuple $(M_a,\nu_a,k_a,L_a)$. 
		\item Let $\lambda = lcm(\{k_a \: | \: a\in [r]\})$. Then set for each $a\in [r]$
		\begin{align*}
			\alpha_a=\lambda/k_a \qquad \text{ and } \qquad R_a=M_a\otimes I_{\alpha_a}
		\end{align*}
			where $I_{\alpha_a}$ is the identity matrix of size $\alpha_a$ and $\otimes$ denotes the Kronecker product.
			Moreover, let 
			\begin{align*}
				L'_a=L_a\otimes \textbf{1}_{\alpha_a} 
			\end{align*}
			where $\textbf{1}_{\alpha_a}$ is the all-one row vector of length $\alpha_a$ and $L'_a$ is therefore the vector obtained by replicating each entry of $L_a$ $\alpha_a$ times.
		 	\item Take a $t\times r$ Vandermonde matrix $A$ over $\BF_q$.

Moreover, let 
			\begin{align*}
				M_T=(A\otimes I_\lambda)diag(R_1,\ldots,R_r)
			\end{align*}
			where $I_\lambda$ is the identity matrix of size $\lambda$ and 
			\begin{align*}
				L_T=concat(L'_a) \text{ for all } a\in [r]
			\end{align*}
			with \emph{concat} denoting the concatenation of all vectors $L_a$ for $a\in[r]$.
Finally, let 
			\begin{align*}
				k_T= \lambda t \qquad \text{ and } \qquad   \nu_T= \sum_{a=1}^r \alpha_a \nu_a
			\end{align*}        

		Note that $k_T$ is the number of rows and $\nu_T$ is the number of columns of $M_T$.	
	\end{enumerate}
\end{enumerate}
\noindent
This algorithm has polynomial complexity in the size of the access tree. 
Note that the field size $q$ needs to be greater or equal than each $r_i$ to allow constructing a $t_i \times r_i$ Vandermonde matrix at every step.
Since each $r_i$ is smaller or equal than $n$, we could also choose the field size as a prime power $q \geq n$. In the general case, this leads to a larger field size, but would not require to visit the tree to compute $q$.

\subsection{Analysis}

\paragraph{Completeness of the resulting linear monotone erasure code.} In the following, we prove that Algorithm~\ref{alg:matrix_construction} provides a linear monotone erasure code for an access structure $\CA$ given as a tree $T$ that is complete. We first state two preliminary lemmas, and then prove the statement in Theorem~\ref{theo: completeness efficient construction}.

\begin{lemma}\label{Lemma:completness_of_construction_1} Let \( M \in \mathbb{F}_q^{k \times \nu} \) and let \( \Gamma \subseteq [\nu] \) be an index set such that the columns of \( M \) indexed by \( \Gamma \) are linearly independent. For any \( s \in \mathbb{N}_{>0} \), define \( R = M \otimes I_s \), where $I_s$ is the identity matrix of size $s$. Then the columns of \( R \) corresponding to \( \Gamma \) are linearly independent.
\end{lemma}
\begin{proof}
	Note that by definition of $R$, a column $a_i$ of $M$ corresponds to $s$ columns of $R$. These columns have the form $a_{i}\otimes e_j$, where $e_j$ is the unit column vector with $1$ at position $j$, for all $j\in [s]$. Now, let $a_{i_1},\ldots,a_{i_\gamma}$ be $\gamma$ columns of $M$ that are linearly independent. 
	Assume that 
	\begin{align*}
		\sum_{\ell=1}^\gamma \sum_{j=1}^{s} c_{\ell,j}(a_{i_\ell}\otimes e_j)=0
	\end{align*}
	where $ c_{\ell,j}\in \BF_q$.
	Note that this can be rewritten as 
	\begin{align*}
		\sum_{\ell=1}^\gamma  a_{i_\ell}\otimes \omega_\ell=0
	\end{align*}
	where $\omega_\ell= \sum_{j=1}^{s} c_{\ell,j} e_j $.
	Since $a_{i_\ell}$ for $\ell \in [\gamma]$ are linearly independent, it follows that $\omega_\ell=0$ for all $\ell\in [\gamma]$ which holds if and only if $c_{\ell,j}=0$. Thus, the statement follows.\\
\end{proof}
\begin{lemma}\label{Lemma:completness_of_construction_2} 
	Let $A\in \BF_q^{t\times t}$ and $R_i\in \BF_q^{s\times s}$ for $i\in [t]$ be invertible. Then the matrix $$M=(A\otimes I_s)diag(R_1,\ldots,R_t)\in \BF_q^{ts\times ts}$$ where $I_s$ is the identity matrix of size $s$ is invertible.
\end{lemma}
\begin{proof}
	Since each $R_i$ is invertible, the block diagonal matrix $$diag({R}_{1},\ldots,{R}_{t})\in \BF_q^{ts
	\times ts}$$ is invertible. Thus $M$ is invertible if and only if $$A\otimes I_s\in \BF_q^{ts \times ts}$$ is invertible. Since the matrix $A\otimes I_s$ is invertible by Lemma~\ref{Lemma:completness_of_construction_1}, the statement follows.
\end{proof}

	\begin{theorem}\label{theo: completeness efficient construction}
		For a given access structure $\mathcal{A}$ represented as an access tree $T$, Algorithm~\ref{alg:matrix_construction} provides a monotone erasure code $\CC$ defined by $M_T$ and $\phi_T$ that is complete.
	\end{theorem}
	\begin{proof} 
	Let $V$ denote the set of vertices in $T$ whose children are leaves. Note that each vertex $v\in V$ correspond to a submatrix $P_v$ of $M_T$, i.e., $P_v$ consists of the columns of $M_T$ that are generated through vertex $v$. Moreover, for a vertex $v \in V$, let $C_v$ be the set of its leaf children. Let
\[
S_v \subseteq C_v
\]
be a subset of leaves of $v$ chosen to satisfy the threshold condition at $v$. Note that $S_v$ corresponds to a submatrix $P_{S_v}$ of $P_v$. 

Now, an access set $A\in \mathcal{A}$ corresponds to a selection of subsets 
of vertices $(S_{v'})_{v'\in V}$ where $S_{v'}\subseteq C_{v'}$. Therefore, to show that $\CC$ is complete it suffices to show that the number of columns of all submatrices $P_{S_{v'}}$ for $v'\in V$ are at least $k$ and that they are jointly linearly independent. 

We first show that these are exactly $k_T$ columns. Without loss of generality, assume that all leaves in $T$ have the same depth. By construction, for a subtree of depth one consisting of a root $v\in V$ labeled $t$ with $r$ children, the resulting matrix $M_v$ has $t$ rows and $r$ columns. Moreover, exactly one column of $M_v$ is assigned to a child of $v$. Thus, a subset of $t$ children of $v$ own together exactly $t$ different columns of $M_v$. Now, consider a subtree of depth two consisting of a root $\tilde{v}$ labeled $\tilde{t}$ with $\tilde{r}$ children $v_1,\ldots,v_{\tilde{r}}$. Moreover, let $t_i$ be the label of each child $v_i$ of $\tilde{v}$. Note that a column of $M_{v_i}$ corresponds to $\lambda$ columns in $R_{v_i}=M_{v_i}\otimes I_\lambda$. Thus, by construction, $\tilde{t}$ children of $\tilde{v}$ own together $\lambda \cdot \tilde{t}$ different columns, which is the number of rows of $M_{\tilde{v}}$. By continuing this procedure for subtrees with depth greater than two, it follows that the number of columns of all submatrices $P_{S_{v'}}$ are exactly $k_T$.

It remains to show that these $k_T$ columns are linearly independent. Note that each subtree of depth one consisting of a root $v'\in V$ labeled $t$ has the property that any $t$ columns are linearly independent. By the following bottom-up construction to yield the matrices $P_{S_{v'}}$ for $v'\in V$ and by Lemma~\ref{Lemma:completness_of_construction_1} and Lemma~\ref{Lemma:completness_of_construction_2} the statement follows.
	\end{proof}

	\paragraph{Size of resulting matrix $M_T$.} Let $T$ be the tree of $\mathcal{A}$ with root $v_\CR$ and let $V=\{v_1, \dots, v_\delta\}$ be the set of vertices whose children are leaves. Moreover, for each vertex $v_i\in V$, let $R_{v_i}$ denote the set of vertices on the unique path from $v_\CR$ to $v_i$. For a vertex $v$, we denote by $x_v$ the label (threshold) of $v$ and by $y_v$ the number of its children in $T$.

 	\begin{theorem}\label{theo:size_of_matrix}
		The resulting matrix $M_T$ by Algorithm~\ref{alg:matrix_construction} has $k$ rows and $\sum_{i=1}^\delta \psi_i \cdot k$  columns where $\psi_i=\frac{y_{v_i}}{\prod_{v\in R_{v_i}}x_v}$ for $i\in [\delta]$ and $k=lcm\left(den(\psi_j)_{j\in [\delta]}\right)$.
	\end{theorem}
	\begin{proof}

Without loss of generality, assume that all leaves in $T$ have the same depth. First note that for a subtree of depth one consisting of a root $v\in V$ labeled $x_v$ with $y_v$ children, Algorithm~\ref{alg:matrix_construction} returns a matrix with $x_{v}$ rows and $y_v$ columns. Moreover, for a subtree of depth two consisting of a root $v$ labeled $t$ with $r$ children $v_1,\ldots,v_r$, the number of rows of the matrix $\tilde{M}$ returned by the algorithm is $$k_v=t\cdot lcm(x_{v_1},\ldots,x_{v_r})= lcm(t\cdot x_{v_1},\ldots,t\cdot x_{v_r}).$$ This is the least common multiple of the products of vertex labels along the paths from $v_i$ to the root $v$, for all $i\in [r]$. 

Moreover, the number of columns of $\tilde{M}$ is $$\nu_v=\sum_{i=1}^r \frac{lcm(x_{v_1},\ldots,x_{v_r})}{x_{v_i}}y_{v_i}=\sum_{i=1}^r \frac{y_{v_i}}{t\cdot x_{v_i}}k_\nu.$$

Now, consider a subtree of depth three labeled $\tilde{t}$ consisting of a root $v$ with $\tilde{r}$ children $u_1,\ldots,u_{\tilde{r}}$. Assume each vertex $u_i$ is labeled $t_i$ and has $r_i$ children and assume that each child of $u_i$ for $i\in [\tilde{r}]$ is labeled with $x_j^i$ and has $y_j^i$ children for $j\in[r_i]$.
Following the same reasoning as for depth two, the number of rows of the resulting matrix by the algorithm is
\begin{align*}
	\tilde{t} \cdot lcm\left (k_{u_1},\ldots,k_{u_{\tilde{r}}}\right)
	&=\tilde{t}\cdot lcm \left( lcm \left(t_1\cdot x^{1}_{1},\ldots,t_1\cdot x^{1}_{{r_1}}\right),\ldots , lcm \left (t_{\tilde{r}}\cdot x^{\tilde{r}}_{1},\ldots,t_{\tilde{r}}\cdot x^{\tilde{r}}_{{r_{\tilde{r}}}} \right)    \right) \\
	&= lcm \left( \tilde{t} \cdot t_1\cdot x^{1}_{1},\ldots, \tilde{t} \cdot t_1\cdot x^{1}_{{r_1}},\ldots , \tilde{t} \cdot t_{\tilde{r}}\cdot x^{\tilde{r}}_{1},\ldots, \tilde{t} \cdot t_{\tilde{r}} \cdot x^{\tilde{r}}_{{r_{\tilde{r}}}}    \right)
\end{align*}

This is again the least common multiple of the products of vertex labels along the paths from a child of $u_i$ to the root $v$, for all $i\in [\tilde{r}]$, $j\in [r_i]$.

On the other hand, the number of columns of the resulting matrix is
\begin{align}\label{proof:size}
	\sum_{i=1}^{\tilde{r}} \frac{lcm\left(k_{u_1},\ldots,k_{u_{\tilde{r}}}\right)}{k_{u_i}}\cdot \nu_{u_i} &= \sum_{i=1}^{\tilde{r}}\frac{lcm \left( \tilde{t}\cdot t_1\cdot x^{1}_{1},\ldots, \tilde{t}\cdot t_1\cdot x^{1}_{{r_1}},\ldots , \tilde{t}\cdot  t_{\tilde{r}}\cdot x^{\tilde{r}}_{1},\ldots,  \tilde{t}\cdot  t_{\tilde{r}} \cdot x^{\tilde{r}}_{{r_{\tilde{r}}}}    \right)}{\tilde{t}\cdot  lcm \left(t_i\cdot x^{i}_{1},\ldots, t_i\cdot x^{i}_{{r_i}}\right)}\cdot \nu_{u_i}
\end{align}

Since $$\nu_{u_i}=\sum_{j=1}^{r_i} \frac{y_{j}^i}{t_i\cdot x_{j}^i}\cdot lcm(t_i\cdot x^i_{1},\ldots,t_i\cdot x^i_{{r_i}})$$

the sum~(\ref{proof:size}) reduces to
\begin{align*}
 \sum_{i=1}^{\tilde{r}}\left( \frac{y_1^i}{\tilde{t}\cdot  t_1\cdot x_1^i}+ \ldots +\frac{y_{r_i}^i}{\tilde{t}\cdot  t_1\cdot x_{r_i}^i}\right)\cdot lcm \left( \tilde{t}\cdot  t_1\cdot x^{1}_{1},\ldots,\tilde{t}\cdot  t_1\cdot x^{1}_{{r_1}},\ldots ,\tilde{t}\cdot   t_{\tilde{r}}\cdot x^{\tilde{r}}_{1},\ldots, \tilde{t}\cdot   t_{\tilde{r}} \cdot x^{\tilde{r}}_{{r_{\tilde{r}}}}    \right).
\end{align*}
Note that since $\delta=r_1+\ldots+r_{\tilde{r}}$, this sum is equivalent to $$\sum_{\ell=1}^\delta \frac{y_{v_\ell}}{\prod_{v\in R_{v_\ell}}x_v} \cdot k$$ where $v_\ell$ are the vertices labeled $x_{v_\ell}=x_j^i$ for $ i\in [\tilde{r}], j\in [r_i]$ and $k=lcm\left( \left (\prod_{v\in R_{v_\ell}}x_v\right)_{\ell\in [\delta]}\right)$.

By continuing this procedure for subtrees with depths greater than three the statement follows.
	\end{proof}

\paragraph{Overhead of resulting linear monotone erasure code.} The overhead of the code $\CC$ constructed by Algorithm~\ref{alg:matrix_construction} is given by the following proposition. 
	 
	  \begin{proposition}
	For a given access structure $\mathcal{A}$ represented as an access tree $T$, Algorithm~\ref{alg:matrix_construction} provides a monotone erasure code $\CC$ defined by $M_T$ and $\phi_T$ with overhead $\beta=\sum_{i=1}^\delta \psi_i -1$, where  $\psi_i=\frac{y_{v_i}}{\prod_{v\in R_{v_i}}x_v}$ for $i\in [\delta]$.
	\end{proposition}
	\begin{proof}
		By Theorem~\ref{theo:size_of_matrix} it follows that $M_T$ has $k$ rows and $m=\sum_{i=1}^\delta \psi_i \cdot k$  columns where $\psi_i=\frac{y_{v_i}}{\prod_{v\in R_{v_i}}x_v}$ for $i\in [\delta]$ and $k=lcm(den(\psi_j)_{j\in [\delta]})$. Therefore, by definition of the overhead it holds that
		
		$$\beta=\frac{m-k}{k}=\frac{\sum_{i=1}^\delta \psi_i \cdot k-k}{k}=\sum_{i=1}^\delta \psi_i-1$$
		which concludes the proof.
	\end{proof}

\section{Linear Monotone Erasure Codes from Threshold Erasure Codes}
\label{sec:construction_from_threshold} 

In this section, we present an algorithm that, given a set of nodes $\mathcal{P}=\{p_1,\ldots, p_n \}$ and an access structure \(\mathcal{A}\) on \(\mathcal{P}\), produces a linear monotone erasure code that is complete for \(\mathcal{A}\) and has an optimal overhead over all linear monotone erasure codes for \CA.

\subsection{Algorithm details}\label{ssec:algorithm_details}

Our scheme generates an $[m,k]$-linear monotone erasure code $\mathcal{C}$ for $\mathcal{A}$ using an $[m,k]$-linear threshold erasure code $\mathcal{C'}$ as a building block. We refer to $\mathcal{C'}$ as the \textit{base code} and  we call the fragments generated by it \textit{base fragments}. In particular, we choose an $[m,k]$-MDS code as a base code.

The setup phase of our scheme consists of two stages. During the first stage, illustrated in Section~\ref{Setting the parameters}, the scheme aims to find the optimal parameters $m$ and $k$ for $\mathcal{C}'$. Then, in the second stage, we use the base code to build our $[m,k]$-linear monotone erasure code. 
Indeed, we know that we can always build an $[m,k]$-MDS code for all optimal parameters $m$ and $k$, and in Section~\ref{Building the code} we show how to derive an $[m,k]$-linear monotone erasure code from it.

We suppose that the MDS code provides the following two functions:

\begin{itemize}
    \item \textsc{MDSEncode}($f$): given $f \in \mathbb{F}_q^k$, it returns a codeword $c \in \mathbb{F}_q^m$;
    \item \textsc{MDSDecode}($u$): given a vector $u$ of length $m$ whose entries are either a symbol or $\bot$, representing an erasure, it returns a vector $f' \in \mathbb{F}_q^k$ or $\bot$.%
\end{itemize}

Once the setup phase is completed, the resulting code can be used to encode any number of files.
The encoding and decoding procedures are described subsequently.

\subsubsection{Setting the parameters}\label{Setting the parameters}

The procedure receives as input a set of nodes $\mathcal{P}$, along with an access structure $\mathcal{A}=\{A_1,\ldots,A_\omega\}$ on $\mathcal{P}$. The goal of this phase is to compute the optimal values for the parameters $m$ and $k$ of the base code, and thus of the resulting linear monotone erasure code~$\mathcal{C}$. %
By optimal, we mean that they minimize the overhead $\beta = \frac{m-k}{k}$ of $\mathcal{C}$ over all parameters and labeling functions of linear monotone erasure codes for \CA.
To compute such parameters, we first solve the following linear programming problem to find an optimal solution $y$. Then, we derive from $y$ the threshold $k$ and the number of fragments $m_i$ assigned to each $p_i$. Finally, $m$ is obtained by summing all the $m_i$ values.

\textbf{LPP:} Find $y=(y_1,\ldots,y_n)\in \mathbb{R}^n$ such that 
\begin{align}\label{Converted LP problem}
    \min &\sum_{i \in [n]}y_i\notag \\
    \text{subject to } \Gamma \cdot y^\top &\geq 1_\omega^\top\tag{LPP} \\ 
    y &\geq 0 \notag
\end{align}

where $\Gamma$ is a binary $\omega \times n$ matrix such that $\Gamma_{ij}=1$ if and only if $p_j\in A_i$, and $1_\omega=(1,\ldots,1)\in \mathbb{N}^\omega$. 
Note that \ref{Converted LP problem} is solvable, as $y_i=1$ for every $i\in [n]$ is always a solution. Since \ref{Converted LP problem} has only rational coefficients, it follows that an optimal solution $y$ lies in $\mathbb{Q}^n$. 

Given such a solution of \ref{Converted LP problem}, we set $k$ as the lcm of the denominators of all the $y_i$'s, where each $y_i$ is expressed as a reduced fraction,
and set $m_i=y_i\cdot k$, for $i \in [n]$. We then compute $m = \sum_{i =1}^n m_i$, obtaining the optimal values for $m$ and $k$. Indeed, notice that $$\beta = \frac{\sum_{i \in [n]} y_i \cdot k - k }{k} = \sum_{i \in [n]} y_i-1 $$
so minimizing $ \sum_{i \in [n]}y_i$ is equivalent to minimizing $\beta$.
Moreover, the vectors $s=(m_1,\ldots,m_n)\in \mathbb{N}^n$ and $\kappa_\omega=(k,\ldots,k)\in \mathbb{N}^\omega$ satisfy the constraint $\Gamma \cdot s^\top \geq \kappa_\omega$. This corresponds to requiring that for all $A\in \mathcal{A}$ the nodes in $A$ hold at least $k$ base fragments together. This way, completeness is guaranteed since each access set has at least $k$ fragments, and thus, by Theorem~\ref{Theorem: linear monotone erasure code sufficient property}, each access set is sufficient.

Notice that there is an infinite number of solutions. Indeed, we can take $k$ as any multiple of the lcm of the denominators of all the $y_i$'s, and the corresponding $m_i$'s would still minimize $\beta+1$. 

\begin{example}
    Let $\mathcal{A}=\{\{a,c,d,e\}, \{b,c,d,e\}, \{a,b,c\},\{a,b,d\}, \{a,b,e\}\}$ be an access structure on $\mathcal{P}= \{a,b,c,d,e \}$. Then, we have
    \begin{align*}
    \Gamma = \begin{pmatrix}
        1 & 0 & 1 & 1 & 1 \\
        0 & 1 & 1 & 1 & 1 \\
        1 & 1 & 1 & 0 & 0 \\
        1 & 1 & 0 & 1 & 0 \\
        1 & 1 & 0 & 0 & 1 
    \end{pmatrix}
    \in \mathbb{F}_2^{5 \times 5}
\end{align*}
    and the optimal solution of the LP problem 
    \begin{align*}
    \min_y \sum_{i=1}^5 y_i\\
    \text{ subject to } \Gamma \cdot y^\top \geq 1^\top_5\\
     y \geq 0
\end{align*}
is $y=(y_1,\ldots,y_5)=(\frac{2}{5},\frac{2}{5},\frac{1}{5},\frac{1}{5},\frac{1}{5})\in \mathbb{Q}^5$. Setting $k=5$ gives $(m_1,\ldots,m_5)=(y_1\cdot k,\ldots,y_5\cdot k)=(2,2,1,1,1)$. This fragment distribution yields the optimal value for the overhead of the corresponding linear erasure code, which is $\frac{m-k}{k}=\frac{2}{5}$.  \hfill $\triangle$
\end{example}

\subsubsection{Building the code}\label{Building the code}

Once we have determined the optimal parameters $m$ and $k$ of $\mathcal{C}$ for $\mathcal{A}$, together with the values $\{m_i\}_{i \in [n]}$, we choose a prime power $q$ such that $q\geq m$. 
We then take an $[m,k]$-MDS code over  $\mathbb{F}_q$ and turn it into an $[m,k]$-linear monotone erasure code by building an appropriate labeling function $\phi$.
In particular, we partition the set $[m]$ into $n$ disjoint subsets $\Omega_1,\ldots,\Omega_n$ of cardinalities $m_1, \ldots, m_n$, where $\Omega_1$ contains the first $m_1$ indices, $\Omega_2$ the indices from $m_1+1$ to $m_1+m_2$, and so on.
Then for each $i \in [m]$, we set
$$\phi(i)= p_j \text{ iff. } i\in \Omega_j.$$

\textbf{Encoding.}
To encode a file $f$, we first give $f$ as input to the function \textsc{MDSEncode}$()$ of the base code, which is part of the given implementation. \textsc{MDSEncode}$(f)$ then returns the base fragments $(c_1,\ldots,c_m) \in \mathbb{F}_q^m$. Finally, we assign to each $p_i$ a list of base fragments $g_i = (c_{j_1}, \ldots, c_{j_{m_i}})$, where $p_i = \phi(j_\ell)$ for $\ell=1, \dots, m_i$. Details of the procedure are given in Algorithm~\ref{alg:encode-bb}.

Note that since we are using a linear MDS code, the base fragments are given by $(c_1,\ldots,c_m) = f \cdot G$, where $G\in \mathbb{F}_q^{k\times m}$ is the generator matrix of the MDS code. Moreover, the fragment assigned to $p_i$ is $g_i =  f \cdot G_{p_i}$. Finally, if $m_i=0$, then the fragment of $p_i$ is $\bot$.

\begin{algo*}
\caption{\textsc{Encode}$(f)$}\label{alg:encode-bb}
\raggedright
\textbf{Input:} file $f \in \mathbb{F}_q^k$, disjoint sets $\Omega_1,\ldots,\Omega_n$ such that $\bigcup_{i =1}^n \Omega_i = [m]$.\\
\textbf{Output:} fragments $(g_1,\ldots,g_n)$.

  \vbox{
  \small
  \begin{numbertabbing}
    xxxx\=xxxx\=xxxx\=xxxx\=xxxx\=xxxx\=MMMMMMMMMMMMMMMMMMM\=\kill
    $(c_1,\ldots,c_m) \xleftarrow[]{} \textsc{MDSEncode}(f)$ \label{}\\
    \textbf{for} $i \in [n]$ \textbf{do} \label{}\\
    \>  $g_i \xleftarrow[]{} [ \; ]$ \label{}\\
    \> \textbf{for} $j \in \Omega_i$ \textbf{do} \label{}\\
    \>\>  $\op{append}(g_i,  c_j)$ \label{}\\
    \> \textbf{if} $g_i = [ \;]$ \textbf{then}\label{}\\
    \>\> $g_i \gets \bot$ \label{}\\
    \textbf{return}  $(g_1, \dots, g_n)$ \label{}
     \end{numbertabbing}
  }
\end{algo*}

\textbf{Decoding.} The decoding procedure of $\mathcal{C}$ corresponds to the decoding procedure \textsc{MDSDecode} of the base code. %

\subsection{Bounds on linear monotone erasure codes} \label{Bounds}
In this section, we give bounds on the total number of base fragments $m$ and the overhead $\beta$ of any linear monotone erasure code $\mathcal{C}$ for $\CA$ built as in Section \ref{ssec:algorithm_details}. In the computation of such bounds for $\mathcal{A}$, a central role is played by the size of the smallest access set, i.e, $\min\{ |A|, A \in \mathcal A\}$, which we denote by $\tau$.

\begin{theorem}\label{Theorem: lower and upper bound on m}
    Let $\mathcal{P}$ be a set of nodes with $|\mathcal{P}|=n$, and let $\mathcal{A}$ be an access structure on \CP. If \CC is a linear monotone erasure code that encodes files of length $k$, then the optimal value for $m$ is bounded by
    $$ k
    \leq m\leq\left\lceil\frac{ k}{\tau}\right\rceil\cdot n.
    $$
\end{theorem}
\begin{proof}
The lower bound is straightforward, as at least $k$ fragments must be stored to make decoding possible.
To prove the upper bound, assume that each node  $p_i\in \mathcal{P}$ receives exactly $\left\lceil\frac{ k}{\tau}\right\rceil$ base fragments. Then, for all $A\in \mathcal{A}$, i.e., $|A|=| \{p_{i_1}, \ldots, p_{i_s}\}|=s$, we have 
$$|g_{i_1}\cup \ldots \cup g_{i_s}|= s \cdot \left\lceil\frac{ k}{\tau} \right\rceil \geq
s\cdot
\frac{ k}{\tau}\geq k$$ 
since $s\geq \tau$. Thus, every set $A\in \mathcal{A}$ is able to reconstruct the information vector. Therefore, $\left\lceil\frac{ k}{\tau}\right\rceil\cdot n$ fragments are sufficient. Moreover, using a higher number of fragments would make the overhead worse, so we have \(m\leq \left\lceil\frac{ k}{\tau}\right\rceil\cdot n\).
\end{proof}

The upper bound of Theorem \ref{Theorem: lower and upper bound on m} means that any set $B\subseteq \mathcal{P}$ with $|B|\geq \tau$ is able to reconstruct the information vector $f$. 
In general, since there is a set $B'\in \mathcal{A}$ such that $|B'|=\tau$, there is at least one node $p\in B'$ that stores at least $\left \lceil\frac{k}{\tau}\right \rceil$ base fragments.
Moreover, we have the following direct consequence of Theorem \ref{Theorem: lower and upper bound on m} that highlights how efficient a linear monotone erasure code for a given access structure $\mathcal{A}$ can be.

\begin{corollary}\label{Corollary: Upper bound on overhead}
    The overhead $\beta$ of an $[m,k]$-linear monotone erasure code for a given access structure~$\mathcal{A}$ is bounded for large $k\in \mathbb{N}$ by
    $$\beta \leq \frac{n-\tau}{\tau}.$$
\end{corollary}

\begin{proof}
    From Theorem \ref{Theorem: lower and upper bound on m} we have $m\leq\left\lceil\frac{ k}{\tau}\right\rceil\cdot n$. Therefore, $$\beta=\frac{m-k}{k}\leq \frac{\left\lceil\frac{ k}{\tau}\right\rceil\cdot n-k}{k}< \frac{\left(\frac{k}{\tau}+1\right)\cdot n -k }{k}=\frac{n}{k}+\frac{n-\tau}{\tau}.$$ Now, if $k$ tends to infinity, we get the result.
\end{proof}
As a consequence of this result, we have that if \(k\) is large enough and $\tau$ is close to $n$, then the overhead is close to zero. In particular, this means that when the size of the smallest access set is close to $n$, we add little redundancy. This is not surprising, as this situation goes toward the threshold case with $k=n$, in which each node stores one base fragment only and the overhead is zero.

\section{Partitioned Access Structures}
\label{sec:partitioned_access}

In the previous section, we have shown how to construct a linear monotone
erasure code~$\mathcal{C}$ for an arbitrary access structure~$\mathcal{A}$
from a threshold erasure code.  This works by solving~\ref{Converted LP
  problem}, which is derived from~\CA, and is efficient if \CA has a compact (e.g., polynomial-size) description in the number of
nodes.  However, there may be exponentially many (in~$n$) access sets in \CA.
In this section, we consider a special kind of access structure called
\emph{partitioned access structure}; such an access structure can be
hierarchically decomposed in the sense that the MBF that represents it corresponds to a tree on the inputs, i.e., on the set of nodes~\CP.
Every node in~\CP appears exactly once in the MBF that describes 
the access structure.

\begin{definition}\label{def:part_as}
  Let $\mathcal{P}$ be a set of nodes. A \emph{partitioned access structure}
  $\mathcal{A}$ on $\mathcal{P}$ is given by an access tree $T$ with root 
  $v_{\mathcal{R}}$, in which every internal vertex~$v$ with $r$ children
  is labeled with $t$ and corresponds to the threshold operator~$\Theta_t^r$.
  The leaves correspond to the nodes in
  $\mathcal{P}$ and are partitioned into $\delta$ sets $B_i$. In particular,
  for each $i\in [\delta]$, there is a uniquely determined vertex $v_i$ such
  that $B_i$ consists exactly of the children of~$v_i$.
\end{definition}

An \emph{$L$-level partitioned access structure} $\mathcal{A}$ is represented
by a tree with $L$ levels.  We denote the number of vertices on
level~$j \in [L]$ by $s_j$, where $s_1=1$ since the root is on level one;
the leaves are on level $L+1$.
In the following we always assume that the access tree is balanced, i.e., that the path from the root to each leaf node contains exactly $L+1$ vertices. Note that an access tree $T$ can always be made balanced by adding vertices with label $1$ to any path shorter than $L+1$.

We focus on partitioned access structures here because they admit an efficient
construction of monotone erasure codes, as shown in the remainder of this
section.

\begin{example}
Consider several organizations (universities, governments, companies, etc.) that each run a set of nodes in a distributed system, and an access structure corresponding to the assumption that more than two thirds of the organizations will remain available and, among each organization that remains available, more than half of the organization's nodes will remain available. The Stellar network~\cite{DBLP:conf/sosp/LokhavaLMHBGJMM19} is a real-world example of a system using this trust model.  Its node set~\CP is partitioned into 
organizations~$B_1, \dots, B_r$, and the access structure can be expressed 
by an MBF like
\[
  \Theta_{t}^r(\Theta_{z_1}^{b_1}(B_1), \dots, \Theta_{z_r}^{b_r}(B_r)),
\]
where the notation $\Theta_x^y(B_i)$ stands for $\Theta_x^y(\{p_j\}_{p_j \in B_i})$.

More generally, one can consider sub-organizations or groups of organizations (e.g. all companies, all universities, etc.) and create hierarchical, partitioned access structures with arbitrarily many levels.
\end{example}

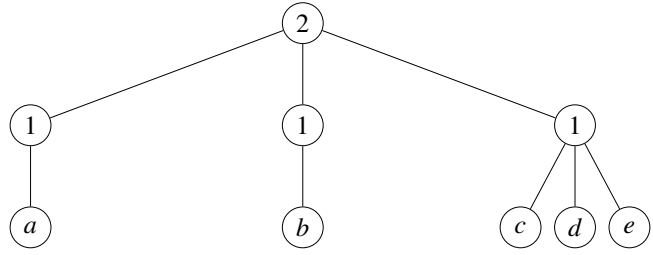
\begin{figure}[h]
    \centering
    \scalebox{0.9}{\begin{tikzpicture}[grow=down, level distance=1.5cm]
    \tikzset{
    circ/.style = {draw, circle, minimum size = 6mm, inner sep = 2pt}}
\node[circ]{$2$} %
  child [sibling distance=4cm]{node[circ]{$1$}
        child [sibling distance=0.8cm]{node[circ]{$a$}}}
  child [sibling distance=4cm]{node[circ]{$1$}
        child [sibling distance=0.8cm]{node[circ]{$b$}}}
  child [sibling distance=4cm]{node[circ]{$1$}
        child [sibling distance=0.8cm]{node[circ]{$c$}}
        child [sibling distance=0.8cm]{node[circ]{$d$}}
        child [sibling distance=0.8cm]{node[circ]{$e$}}}
 ;
\end{tikzpicture}}
\caption{Access tree of Example \ref{example: uniform assignment is not necessarily optimal}.}
\label{Figure: Reduced access tree 1}
\end{figure}

\subsection{Uniform assignment}\label{sec:uniform_assignment}

Let $\mathcal{A}$ be an $L$-level partitioned access structure with tree $T$ and partition $B_1, \dots, B_\delta$. %
We propose here a fragments assignment, called \emph{uniform assignment}, that is efficient to compute but not always optimal in terms of overhead. In more detail, the uniform assignment provides parameters \((m, k)\) for the base code such that the resulting monotone erasure code is complete. 

Let $R_v$ be the set of vertices on the unique path from the root to vertex $v$. For all $i\in [\delta]$, the number of base fragments assigned to each $p_j\in B_i$ is
$$m_j=\frac{k}{\prod\limits_{v\in R_{v_i}}x_v},$$ 
where $v_i$ is the vertex whose set of children is $B_i$ and $k=lcm(den(m_j)_{j\in [n]})$. Figure~\ref{Illustration: uniform fragment distribution for special access structures} shows a visual representation of the uniform assignment on a two-level partitioned access structure. The idea is that each child of a vertex in the tree gets the same amount of information. In more detail, if the file has size $k$, then any child of the root gets an information of size $k/t$. In this way, any $k$ children of the root have enough information for reconstruction. With the same reasoning, the information is further distributed among all paths until we reach the leaves. This ensures that any set $A\in \mathcal{A}$ holds exactly $k$ fragments, and thus it is able to reconstruct the file. Intuitively, giving less fragments to some $p \in \CP$ makes some access set insufficient, which violates completeness of the code.

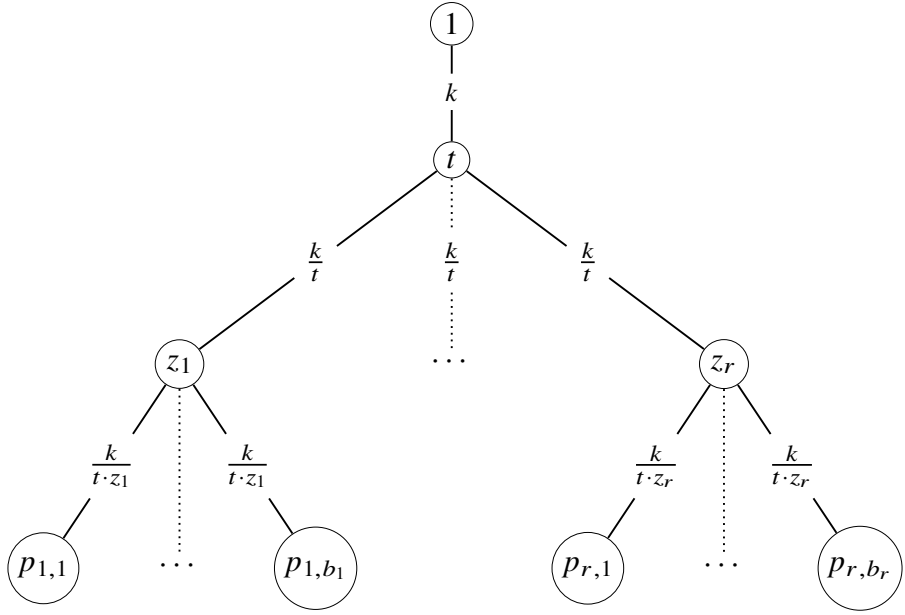
\begin{figure}[h]
\begin{center}
    \scalebox{0.9}{\begin{tikzpicture}
    \tikzset{
    circ/.style = {draw, circle, inner sep=2pt}}
        \node[circ, scale=1.2] (O) at (0,2) {$1$};

        \node (K) at (0,1) {$k$};
        
        \node[circ, scale=1.2] (a) at (0,0) {$t$};

        \node[circ, scale=1.2] (b) at (-4,-3) {$z_{1}$};
        \node[scale=1.2] (b1) at (-2,-1.5) {$\frac{k}{t}$};
        \node[scale=1.2] (e) at (0,-3) {$\cdots$};
        \node[scale=1.2] (e1) at (0,-1.5) {$\frac{k}{t}$};
        \node[circ, scale=1.2] (g) at (4,-3) {$z_r$};
        \node[scale=1.2] (g1) at (2,-1.5) {$\frac{k}{t}$};

        \node[circ, scale=1.2] (h) at (-6,-6) {$p_{1,1}$};
        \node[scale=1.2] (h1) at (-5,-4.5) {$\frac{k}{t \cdot z_{1}}$};
        \node[scale=1.2] (j) at (-4,-6) {$\cdots$};
        \node[circ, scale=1.2] (k) at (-2,-6) {$p_{1,b_1}$};
        \node[scale=1.2] (k1) at (-3,-4.5) {$\frac{k}{t \cdot z_1}$};

        \node[circ, scale=1.2] (p) at (2,-6) {$p_{r,1}$};
        \node[scale=1.2] (p1) at (3,-4.5) {$\frac{k}{t \cdot z_r}$};
        \node[scale=1.2] (r) at (4,-6) {$\cdots$};
        \node[circ, scale=1.2] (s) at (6,-6) {$p_{r,b_r}$};
        \node[scale=1.2] (s1) at (5,-4.5) {$\frac{k}{t \cdot z_r}$};

        \draw[thick] (O) -- (K) -- (a);
        
        \draw[thick] (a) -- (b1) -- (b);
        \draw[thick,dotted] (a) -- (e1) -- (e);
        \draw[thick] (a) -- (g1) -- (g);

        \draw[thick] (b) -- (h1) -- (h);
        \draw[thick,dotted] (b) -- (j);
        \draw[thick] (b) -- (k1) -- (k);

        \draw[thick] (g) -- (p1) -- (p);
        \draw[thick,dotted] (g) -- (r);
        \draw[thick] (g) -- (s1) -- (s);
    \end{tikzpicture}}
\caption{Uniform assignment of base fragments.}
\label{Illustration: uniform fragment distribution for special access structures}
\end{center}
\end{figure}

\begin{example}
\label{example: uniform assignment is not necessarily optimal}
Consider the following two-level partitioned access structure 
$$\mathcal{A}=\Theta_2^3(\Theta_1^1(a),\Theta_1^1(b),\Theta_1^3(c,d,e)).$$ Figure \ref{Figure: Reduced access tree 1} shows the tree of $\mathcal{A}$. With the uniform assignment to $\mathcal{A}$, each node receives $\frac{k}{2}$ fragments, and this leads to an overhead of $\beta=\frac{3}{2}$. Note that nodes in $\{c,d,e\}$ receive more base fragments together compared to $a$ and $b$. Therefore, the question arises whether the overhead gets smaller if we reduce the number of base fragments given to nodes in $\{c,d,e\}$ and assign more fragments to $a$ and $b$. For instance, if each node in $\{c,d,e\}$ receives no base fragment, we can consider $c,d$ and $e$ as absent. This leads to a new access structure $\mathcal{A}_0=\Theta_1^2(\Theta_1^1(a),\Theta_1^1(b)) = \Theta_1^2(a,b)$.
Note that for each $A \in \mathcal{A}$ there exists a set $A_0\in \mathcal{A}_0$ such that $A_0 \subset A$. This implies that $\mathcal{A}_0$ contains $\mathcal{A}$.
Therefore, $a$ and $b$ must compensate the loss of the base fragments of the nodes in $\{c,d,e\}$ in order to make each access set $A\in \mathcal{A}$ qualified. 
To achieve this, we apply the uniform assignment to $\mathcal{A}_0$, which allocates $k$ fragments to both $a$ and $b$. 
This leads to the overhead $\beta_0=1<\beta$, which shows that the uniform assignment is not optimal for $\mathcal{A}$. 
\hfill $\triangle$
\end{example}
Example~\ref{example: uniform assignment is not necessarily optimal} demonstrates that in some cases it is beneficial to disregard nodes from~$\mathcal{P}$, i.e., to assign them zero base fragments. Based on this, in Section~\ref{ssec:algorithm_partitioned}, we provide a method that builds an optimal assignment for any partitioned access structure.

\subsection{Optimal fragment assignment algorithm}
\label{ssec:algorithm_partitioned}

We present here an algorithm that finds an optimal solution to~\ref{Converted LP problem} for any partitioned access structure in a way that substantially reduces the computational effort compared to solving the minimization problem directly. 
The algorithm takes an access tree \(T\) of a partitioned access
structure as input, and returns the optimal parameters \((\nu_T, k_T, h_T)\) for the base code, where:

\begin{itemize}
    \item $\nu_T$ is total number of base fragments held by all nodes in $T$;
    \item $k_T$ is the reconstruction threshold of the code;
    \item $h_T: [n] \to \BN$ is a function that returns the number of base
      fragments assigned to node~$i \in [n]$ by the construction.
\end{itemize}
We now give an overview of the algorithm. 
Recall that each vertex in $T$ with $r$ children is labeled with a value $t$ representing the threshold of the corresponding operator $\Theta_t^r$.

\begin{enumerate}
\item If \(T\) has depth zero and consists of only one (leaf) vertex~$i$, then
  let $(\nu_T, k_T) = (1, 1)$, let $h_T(i) = 1$, and let $h_T(j) = 0$ for
  $j \neq i$. Return \((\nu_T, k_T, h_T)\).
\item If \(T\) has depth one or more and consists of a root labeled \(t\) and
  \(r\) subtrees \(T_1,\dots,T_r\):
  \begin{enumerate}  \item Apply the algorithm to each subtree \(T_a\) for $a \in [r]$, to obtain
    $(\nu_a, k_a, h_a)$.  Let \(\rho_a=\nu_a/k_a\).  W.l.o.g., assume the
    subtrees are ordered so that if \(a<b\) then \(\rho_a\le \rho_b\).
  \item Consider any $s\in \{r-t+1,\ldots,r-1\}$ such that
    \[
      \rho_{s+1}\geq \frac{1}{s-r+t}\sum_{a=1}^{s} \rho_a
    \]
    and let $s^*$ be the smallest such~$s$ or let $s^* = r$ if no suitable $s$
    satisfies the condition.
    Note that if the condition on $\rho_s$ holds for some $s<r$, then it also holds for any $s'>s$. 
    Let 
    \[
      S=\{1,\ldots,s^*\}
    \]
    and
    \[
      \lambda=\operatorname{lcm}(\{\,k_a \mid a\in S\,\}).
    \]
  \item For each \(a\in S\), set \(\alpha_a=\lambda/k_a\) (an
    integer).  Then rescale the base fragment counts $\nu_a$ and
    $h_a$ returned for each subtree $a \in [r]$ as
    \[
      \nu'_a =
      \begin{cases}
        \alpha_a \nu_a & \text{ if } a \in S \\
        0 & \text{ if } a \not\in S
      \end{cases}
    \]
    Analogously, for each $a \in [r]$ and $i \in [n]$ set
    \[
      h'_a(i) =
      \begin{cases}
        \alpha_a h_a(i) & \text{ if } a \in S \\
        0 & \text{ if } a \not\in S
      \end{cases}
    \]
  \item Finally, let
    \begin{align*}
      \nu_T &=\sum_{a=1}^r \nu'_a,\\
      k_T &=(s^*-r+t)\,\lambda
    \end{align*}
    and the function $h_T: [n] \to \BN$ be such that for $i \in [n]$
    \[
      h_T(i) =
      \begin{cases}
        h'_a(i) & \text{ if there is some $a \in S$ with $h_a(i) > 0$}\\
        0 & \text{ otherwise. }
      \end{cases}
    \]
    For the definition of $h_T$, recall that \CA is partitioned, hence, every
    node appears in at most one subtree.
  \end{enumerate}
\item Return $(\nu_T, k_T, h_T)$.
\end{enumerate}
The resulting code~\CC therefore is a $[\nu_T, k_T]$-linear monotone erasure
code such that the fragment of each node~$i$ in \CC consists of $h_T(i)$
(base) fragments of~$\CC'$. Details of the procedure can be
found in function FA($T,L$) in Algorithm~\ref{alg:fragments_distribution_multi-level}, which takes as input an access tree $T$ with $L$ levels. We assume that each leaf is assigned a global index $i \in [n]$, where~$n$ is globally known.

The time complexity of the algorithm is $\CO(n^2)$, which is considerably lower than the complexity required to directly solve~\ref{Converted LP problem}.

Note that for partitioned access structures, the algorithm discussed in Section~\ref{alg:matrix_construction} can be extended with the algorithm presented in this section. In this case, the resulting linear monotone erasure code is optimal. As the encoding matrix has a block-wise MDS property, this leads to a smaller required field size.

\begin{algo*}[htb]
\caption{\textsc{Fragments assignment}}
  \label{alg:fragments_distribution_multi-level}
\setcounter{g@linecounter}{0}
\raggedright
// $n$ is a global parameter indicating the number of nodes in \CP  

  \vbox{
  \small
  \begin{numbertabbing}
    xxxx\=xxxx\=xxxx\=xxxx\=xxxx\=xxxx\=MMMMMMMMMMMMMMMMMMM\=\kill 
    \textbf{function} FA($T,L$)\label{}\\
    \>$T_1, \dots, T_r \gets \text{subtrees of } T$ \label{}\\
    \>$t \gets \text{label of root } v \text{ of } T$ \label{}\\
    \>$h_T \gets [\;]$:  hash map from $[n]$ to \BN, initially $h_T[z] = 0$ for $z \in [n]$ \label{}\\
    \>$h_i \gets [\;]$ for $i \in [r]$:  hash maps from $[n]$ to \BN, initially $h_i[z] = 0$ for $z \in [n]$ \label{}\\
    \>\textbf{if} $L=0$ \textbf{then} \`// base case where $T$ is the leaf $v_i$ \label{}\\
    \>\> \textbf{for} $j \in [n]$ \textbf{do} \label{}\\
    \>\> \> $h_T[j] \gets 0$ \label{}\\
    \>\> $h_T[i] \gets 1$ \label{}\\
    \>\> $(\nu_T, k_T) \gets (1,1)$ \label{}\\
    \> \textbf{else}\label{}\\
    \>\> \textbf{for} $v_i$ child of $v$ \textbf{do} \label{}\\
    \>\>\> $(\nu_i, k_i, h_i) \gets \text{FA}(T_i, L-1)$ \label{} \\
    \>\>\> $\rho_i \gets \frac{\nu_i}{k_i}$ \label{} \\
    \>\> $\pi \gets [\rho_i]_{i \in [r]}$ \label{} \\
    \>\> $\pi \gets \text{SORT}(\pi) $ \label{} \\
    \>\> $I \gets \left\{s\in \{r-t+1,\ldots,r-1\} \mid
    \rho_{{s+1}}\geq \frac{1}{s-r+t}\sum_{i=1}^{s} \rho_{i}\right\}$ \label{} \\
    \>\> $s^* \gets \min\{I \cup \{r\}\}$ \label{} \\
    \>\> $\lambda \gets \text{lcm}(\{k_i \mid i \leq s^*\}) $\label{}\\
    \>\> \textbf{for} $i \in [s^*]$ \textbf{do} \label{}\\
    \>\>\> $\alpha_i\gets \frac{\lambda}{k_i} $ \label{}\\
    \>\>\> $ \nu_i \gets \alpha_i \cdot  \nu_i $ \label{}\\
    \>\>\> \textbf{for} $ z \in [n]$ \textbf{do} \label{}\\
    \> \>\> \> $h_i[z]\gets \alpha_ih_i[z] $ \label{} \\
    \>\>\textbf{for} $i = s^*+1, \dots, r$ \textbf{do} \label{}\\
    \> \>\> $\nu_i\gets 0$ \label{}\\
    \> \>\>\textbf{for} $ z \in [n]$ \textbf{do}\label{}\\
    \> \>\>\> $h_i[z]\gets 0 $ \label{}\\
    \>\>\textbf{for} $z \in [n]$ \textbf{do} \label{}\\
    \>\>\> \textbf{if} $\exists i \in [r]$ s.t. $h_i[z]>0$ \textbf{then} \label{}\\ 
    \> \>\>\> $h_T [z]\gets h_i[z]$ \label{}\\
    \> \>\>\textbf{else} \label{}\\ 
    \> \>\>\> $h_T[z]\gets 0$ \label{}\\
    \>\>$\nu_T=\sum_{i=1}^r \nu_i$ \label{}\\
    \>\>$k_T=(s^*-r+t)\cdot\lambda$ \label{}\\
    \>\textbf{return} $(\nu_T, k_T, h_T)$ \label{}
    \end{numbertabbing}
  }
\end{algo*}

\begin{example}\label{example: general partitioned access structure}
Consider the following partitioned access structure $$\mathcal{A}=\Theta_2^3(\Theta_3^3(p_1,p_2,p_3),\Theta_2^3(p_4,p_5,p_6),\Theta_1^3(p_7,p_8,p_9)).$$
Figure \ref{fig:access_tree_example_algorithm_partitioned} shows the access tree $T$ of $\mathcal{A}$. Let $T^{(1)},T^{(2)}$ and $T^{(3)}$ be the depth one subtrees of $T$. Note that each subtree $T^{(i)}$ is further decomposed into multiple subtrees of depth zero, each consisting of only one leaf. For simplicity, we identify each of these subtrees with its root node $p_i$. Then, for each leaf $p_i$, the algorithm first delivers the following tuples 
    \begin{align*}
        (\nu_i,k_i,h_i)=\text{FA}(p_i,0)=(1,1,h_i) \text{ where }\\
        h_i(z)=\begin{cases}
            1 & \text{for } z=i\\
            0 &\text{else}
        \end{cases}
    \end{align*}
Let $T$ be the subtree corresponding to $T_{1}^{(1)}$. We get $\nu_T=1, k_T=1$, and 
    \begin{align*}
    h_T(z)=
        \begin{cases}
            1 &\text{for }z=1\\
            0 &\text{else}
        \end{cases}
    \end{align*}
    Thus, $\text{FA}\left(T_{1}^{(1)},1\right)=(1,1,h_1)$, where $h_1=h_T$. Analogously, we get $\text{FA}\left(T_{2}^{(1)},1\right)=(3,2,h_2)$ and $\text{FA}\left(T_{3}^{(1)},1\right)=(3,1,h_3)$ where
        \begin{align*}
        h_2(z)&=\begin{cases}
            1 &\text{if } z=4,5,6\\
            0 &\text{else}
        \end{cases}\\
        h_3(z)&=\begin{cases}
            1 &\text{if } z=7,8,9\\
            0 &\text{else}
        \end{cases}
    \end{align*}
Finally, the output is $\text{FA}(T,3)=(5,2,h_T)$, where
    \begin{align*}
    h_T(z)=
        \begin{cases}
            2 &\text{if } z=1\\
            1 &\text{if } z=4,5,6\\ 
            0 &\text{else}
        \end{cases}
    \end{align*}

    \begin{figure}[h]
    \centering
    \scalebox{0.9}{\begin{tikzpicture}[grow=down, level distance=2cm]
    \tikzset{
    circ/.style = {
        draw,
        circle,
        inner sep = 2pt}}
        
\node[circ]{$2$} %
  child [sibling distance=5cm]{node[circ]{$3$}
                   child [sibling distance=1cm]{node[circ]{$p_{1}$}}
            child [sibling distance=1cm]{node[circ]{$p_{2}$}}
           child [sibling distance=1cm]{node[circ]{$p_{3}$}} }
  child [sibling distance=5cm]{node[circ]{$2$}
            child [sibling distance=1cm]{node[circ]{$p_{4}$}}
            child [sibling distance=1cm]{node[circ]{$p_{5}$}}
            child [sibling distance=1cm]{node[circ]{$p_{6}$}}}
            child [sibling distance=5cm]{node[circ]{$1$}
            child [sibling distance=1cm]{node[circ]{$p_{7}$}}
            child [sibling distance=1cm]{node[circ]{$p_{8}$}}
            child [sibling distance=1cm]{node[circ]{$p_{9}$}}}
 ;
\end{tikzpicture}}
\caption{Access tree $T$ of $\mathcal{A}$}\label{fig:access_tree_example_algorithm_partitioned}

\end{figure}
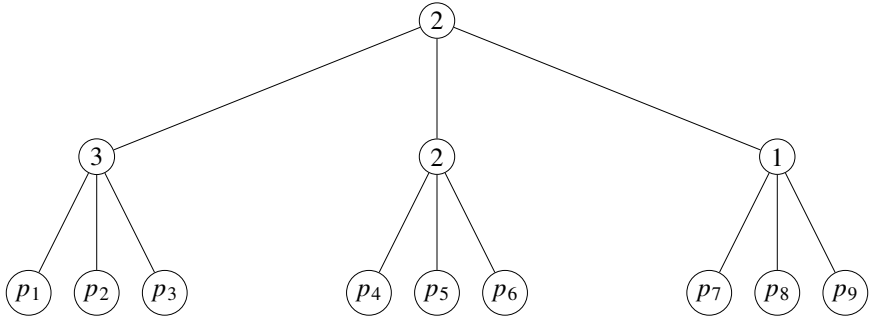
\hfill $\triangle$
\end{example}

\subsection{Analysis}\label{sec:fragments_distribution_analysis}
We now show that applying Algorithm~\ref{alg:fragments_distribution_multi-level} to any partitioned access structure \CA provides an optimal %
linear monotone erasure code for \CA.  To do that, we first rephrase the problem we want to solve, namely, \ref{Converted LP problem}, in a recursive way, which makes the proof easier.

Recall that, given an access structure $\mathcal{A} = \{A_1, \dots, A_\omega\}$, LPP is defined as follows:
\begin{align}\label{Converted LP problem_2}
    \min_{y\in \BR^n} &\sum_{i \in [n]}y_i \notag \\
    \text{subject to } \Gamma \cdot y^\top &\geq 1_\omega^\top\tag{LPP} \\ 
    y &\geq 0 \notag
\end{align}
where $\Gamma$ is a binary $\omega \times n$ matrix such that $\Gamma_{ij}=1$ if and only if $p_j\in A_i$, and $\Gamma_{ij}=0$ otherwise, and $1_\omega=(1,\ldots,1)\in \mathbb{N}^\omega$.
Consider a partitioned access structure \CA with access tree $T$, where the root is labeled $t$ and there are $r$ subtrees $T_1, \dots, T_r$. 
Let us denote by $\CA_\ell$ the access structure determined by subtree $T_\ell$.
Notice that, for all access sets $A \in \CA$, there exist sets $S_1, \dots, S_t$ and a permutation \(\sigma\) of \(\{1, \dots, r\}\) such that $$A=\bigcup_{i \in [t]} S_i$$
where, for each $\ell \in [t]$,  $S_\ell \in \mathcal{A}_{\sigma(\ell)}$. %
For simplicity, assume $\sigma(\ell)=\ell$ for each $\ell \in [t]$. Additionally, let $\CP_\ell$ be the set of nodes corresponding to the leaves in each subtree $T_\ell$.

As required by the first inequality in \ref{Converted LP problem_2}, for each $A \in \CA$ it must hold
\begin{align*}
    \sum_{p_i \in A} y_i  \geq 1 \Longleftrightarrow \sum_{\ell \in [t]} \sum_{p_i \in S_\ell} y_i  \geq 1.
\end{align*}

Let $y' \in \mathbb{R}_{\geq 0}^n$ and $\gamma' \in \mathbb{R}_{\geq 0}^r$ be such that $y_i = \gamma'_\ell y'_i$ for all $p_i \in S_\ell$, with $\ell \in [r]$. Then, we can rewrite~\ref{Converted LP problem_2} as follows
\begin{align}\label{intermediate_problem}
    \min_{\gamma',y'} &\sum_{\ell\in [r]} \gamma'_{\ell}\sum_{p_i\in \mathcal{P}_\ell} y_i' \notag\\ 
    \text{ subject to } &\sum_{\ell\in J}  \gamma'_\ell\sum_{p_i\in S} y_i' \geq 1 \qquad \text{ for all } J \subseteq [r], |J| = t, S \in \CA_\ell \tag{$P'$}\\
    & \gamma', y' \geq  0 \notag
\end{align}

Now, define $d_\ell$ as $\min\{\sum_{p_i\in S} y_i' \: | \: (\gamma', y') \text{ is a solution to~\ref{intermediate_problem} and } S \in \CA_\ell\}$.
Then, 
for each valid $y'$ and for each $S \in \CA_\ell$ it holds
$$
\sum_{p_i\in S} y_i' \geq d_\ell \Longleftrightarrow \sum_{p_i\in S} y_i'' \geq 1
$$
where $y''_i= \frac{y'_i}{d_\ell}$ for all $p_i \in S$. 
We can therefore derive the following subproblem for each $\ell \in [r]$:
\begin{align}\label{P_ell_problem}
    \min_{y''} &\sum_{p_i\in \mathcal{P}_{\ell}} y''_i \notag\\
    \text{ subject to } &\sum_{p_i \in S} y''_i  \geq 1 \qquad \text{ for all $S\in \CA_\ell$}\tag{$P_\ell$} \\ 
    & \qquad y'' \geq 0 \notag
\end{align}
which corresponds to~\ref{Converted LP problem_2} instantiated on $T_\ell$.
As explained in Section~\ref{Setting the parameters}, the threshold of subproblem~(\ref{P_ell_problem}) is $k_{\ell} = lcm(den(y''_i))_{p_i \in \CP_{\ell}}$, and the number of fragments assigned to each $p_i$ is $m_i= y''_i k_\ell$. In particular, this means that $$\sum_{p_i \in \CP_\ell} y''_i \cdot k_\ell= \nu_{\ell}$$
which leads to the following chain of equalities
\begin{align*}
    & \sum_{i \in [n]} y_i = \sum_{\ell\in [r]}\gamma'_\ell\sum_{p_i \in \CP_\ell} y'_i
    = \sum_{\ell \in [r]}\gamma'_\ell d_\ell \sum_{p_i \in \CP_{\ell}}y''_i =
    \sum_{\ell \in [r]}\gamma'_\ell d_\ell  \sum_{p_i \in \CP_{\ell}} \frac{y''_i \cdot k_\ell}{k_{\ell}}
    =\sum_{\ell \in [r]} \gamma'_\ell d_\ell\frac{\nu_\ell}{k_{\ell}} 
\end{align*}
Moreover, if $\gamma_\ell = \gamma'_\ell d_\ell$, then the constraint in Problem~(\ref{intermediate_problem}) is equivalent to
$$
\sum_{\ell\in J}  \gamma_\ell\sum_{p_i\in S} y_i'' \geq 1
$$
which is equivalent to
$\sum_{\ell\in J}  \gamma_\ell \geq 1$
as $\sum_{p_i\in S} y_i'' \geq 1$.%
We can therefore rewrite \ref{Converted LP problem_2} as follows: 
\begin{align*}\label{LPP:to_solve}
    \min_{\gamma,\psi} &\sum_{\ell=1}^r \gamma_\ell \frac{\nu_\ell}{k_\ell}\\
        \text{ subject to }&\sum_{j\in J}\gamma_j\geq 1 \qquad \text{ for all $J\subseteq [r]$, $|J|=t$}\tag{$P$} \\
         & \qquad \gamma \geq  0 
\end{align*}
where $\psi=(\frac{\nu_\ell}{k_\ell})_{\ell \in [r]}$, with $\frac{\nu_\ell}{k_\ell}$ given by a valid (not necessarily optimal) solution of subproblem~(\ref{P_ell_problem}) instantiated on $T_\ell$, for each $\ell \in [r]$.

We are now ready to state some preliminary lemmas that will be used in the proof of
Theorem~\ref{theo:optimal_parameters}, which ensures that Algorithm~\ref{alg:fragments_distribution_multi-level} provides the optimal parameters of a linear monotone erasure code for any partitioned access structure.
Our approach consists in showing that, at each step, the algorithm finds an optimal solution of a subproblem, and then it combines such solutions optimally to build a solution of the main problem.
In particular, Lemma~\ref{lemma:cost_preserving} shows that an optimal solution on the main problem induces optimal solutions on all subproblems. Then, Lemma~\ref{lemma:bins} shows how to find an optimal solution of a subproblem.

\begin{lemma}\label{lemma:cost_preserving}
    Let $T$ be the access tree of a partitioned access structure consisting of a root labeled $t$ and $r$ subtrees $T_1,\ldots,T_r$. Then, every optimal solution of~(\ref{LPP:to_solve}) induces an optimal solution of each subproblem~(\ref{P_ell_problem}) for $\ell\in [r]$. 
\end{lemma}
\begin{proof} %

Assume we have a valid solution $(\gamma,\psi)$ to~(\ref{LPP:to_solve}) corresponding to the objective value $\sum_{\ell=1}^r\gamma_\ell \frac{\nu_{\ell}}{k_{\ell}}$. Then, assume there exists a valid solution $(\gamma,\tilde{\psi})$ to~(\ref{LPP:to_solve}) such that there is some $q\in [r]$  with $\frac{\tilde{\nu}_{q}}{\tilde{k}_{q}}\leq \frac{\nu_{q}}{k_{q}}$, and $\frac{\tilde{\nu}_{\ell}}{\tilde{k}_{\ell}}= \frac{\nu_{\ell}}{k_{\ell}}$ for $\ell \neq q$. Then we have $$\sum_{i=1, i\neq q}^r \gamma_i \frac{\nu_{i}}{k_{i}} + \gamma_q\frac{\tilde{\nu}_{q}}{\tilde{k}_{q}} \leq \sum_{i=1}^r\gamma_i \frac{\nu_{i}}{k_{i}}.$$ Therefore, minimizing $\frac{\nu_q}{k_q}$ reduces the objective value of~(\ref{LPP:to_solve}), thus, an optimal solution to~(\ref{LPP:to_solve}) must minimize $\frac{\nu_{i}}{k_{i}}$ for each $i\in [r]$, which concludes the proof.
\end{proof} 
As a consequence of Lemma \ref{lemma:cost_preserving}, (\ref{LPP:to_solve}) reduces to 
\begin{align*}
\label{LPP:to_solve_reduced}
    \min_{\gamma} &\sum_{\ell=1}^r \gamma_\ell \frac{\nu^*_\ell}{k^*_\ell}\\
        \text{subject to }&\sum_{j\in J}\gamma_j\geq 1 \qquad \text{ for all $J\subseteq [r]$, $|J|=t$}\tag{$P^*$}\\
        & \qquad \gamma \geq  0 
\end{align*}
where $\frac{\nu^*_\ell}{k^*_\ell}$ is given by an optimal solution
of subproblem~(\ref{P_ell_problem}) instantiated on $T_\ell$, for each $\ell \in [r]$. Therefore, the solution to each subproblem~(\ref{P_ell_problem}) can be seen as a constant $c_\ell$. The next lemma shows how to find an optimal solution to (\ref{LPP:to_solve_reduced}).
\begin{lemma}\label{lemma:bins}
Let $r,t$ be integers such that $t\leq r$, and let $c=(c_i)_{i\in [r]}$, where $c_i\in \BQ_{\geq 1}$ such that $c_i\leq c_j$ for $i < j$. Moreover, let 
$$s^* = \min\left\{s\in \{r-t+1,\ldots,r-1\}\mid
                      c_{s+1}\geq \frac{1}{s-r+t}\sum_{i=1}^{s} c_i\right\}.$$
If it does not exist, set $s^*=r$. Then, the optimal solution of the following problem
\begin{align}
    \min_{\gamma} &\sum_{i=1}^r \gamma_i c_i \notag\\
    &\sum_{j\in J} \gamma_j\geq 1 \text{ for all } J\subseteq [r] \text{ with } |J|=t \tag{$P(c)$}\\
    &\qquad \gamma\geq 0 \notag
\end{align}
where $\gamma_i\in \BR_{\geq0}$ is given by $\gamma_{i}=0$ for all $i\geq s^*+1$ and $\gamma_j=\frac{1}{s^*-r+t}$ for all $j\leq s^*$. 

\end{lemma}
\begin{proof}
    First, note that under the conditions $c_i\leq c_j$ and $\gamma_i\leq  \gamma_j$, where $i\neq j\in [r]$, swapping $\gamma_j$ and $\gamma_i$ decreases the objective value and does not violate validity of the solution. Therefore, an optimal solution satisfies
    $$c_1\leq \ldots \leq c_r \text{ and } \gamma_1\geq \ldots \geq \gamma_r.$$
    Moreover, note that for a solution to be valid, it must hold that $\gamma_{r-t+1} + ... + \gamma_r \geq 1$. This means that we can set $\gamma_1 = \ldots = \gamma_{r-t+1}$ because it decreases the objective value without violating validity. Furthermore, if $\gamma_{r-t+1} +\ldots + \gamma_r=v>1$, then setting $\tilde{\gamma}_i=\frac{\gamma_i}{v}$ for $r-t+1 \leq i\leq r$ leads to a smaller objective value without violating validity.
    Therefore, we can assume that $\gamma_{r-t+1} + \ldots + \gamma_r=1$. If $s^*<r$, we can reduce $\gamma_j$ with $j> s^*$ and increase $\gamma_i$ with $i \in [s^*]$. In particular, setting $\tilde{\gamma}_i=\gamma_i+\sum_{j=s^*+1}^r\frac{\gamma_{j}}{s^*-r+t}$, for all $i\in [s^*]$, and $\tilde{\gamma}_j=0$, for all $j\geq s^*+1$, decreases the objective value. Indeed, we have
    \begin{align*}
        \sum_{i=1}^{s^*}\tilde{\gamma}_ic_i=\sum_{i=1}^{s^*}\left[\left( \gamma_i + \sum_{j=s^*+1}^r \frac{\gamma_j}{s^*-r+t}\right)c_i\right] &= \sum_{i=1}^{s^*}\gamma_i c_i + \sum_{i=1}^{s^*} \left[ \frac{c_i}{s^*-r+t}\cdot \sum_{j=s^*+1}^r\gamma_j \right] \\
        &= \sum_{i=1}^{s^*}\gamma_i c_i + \sum_{j=s^*+1}^r\gamma_j\cdot \left(\frac{1}{s^*-r+t}\cdot \sum_{i=1}^{s^*}c_i\right)\\
        &\leq \sum_{i=1}^{s^*}\gamma_i c_i + \sum_{j=s^*+1}^r \gamma_j c_j = \sum_{i=1}^r\gamma_i c_i
    \end{align*}
since $\frac{1}{s^*-r+t}\cdot \sum_{i=1}^{s^*}c_i \leq c_j$ for all $j\in \{s^*+1,\ldots,r\}$ and $\gamma_j\geq 0$. Moreover, this step does not violate validity. Indeed, if $J\subseteq [s^*]$ with $|J|=s^*-r+t$, then
\begin{align*}
    \sum_{i\in J}\tilde{\gamma}_i=\sum_{i \in J} \left[\gamma_i + \sum_{j=s^*+1}^r \frac{\gamma_j}{s^*-r+t} \right] &=  \sum_{i \in J} \gamma_i + (s^*-r+t) \sum_{j=s^*+1}^r \frac{\gamma_j}{s^*-r+t} = \sum_{i \in J} \gamma_i + \sum_{j=s^*+1}^r \gamma_j \geq 1
\end{align*}
Thus, we can assume that $\gamma_r=\ldots=\gamma_{s^*+1}=0$. Otherwise, if $s^*=r$, then it must hold $\gamma_i>0$ for all $i\in [r]$.

Similarly to how we reasoned to set some values of $\gamma$ zero, observe that if $\gamma_{r-t+2}\neq 0$, it is beneficial to set $\tilde{\gamma}_\ell=\gamma_\ell -\epsilon$ with $\epsilon>0$ for all $\ell \in [r-t+1]$, and $\tilde{\gamma}_{r-t+2}=\gamma_{r-t+2}+\epsilon$. Indeed,  
\begin{align*}
     \sum_{i=1}^{s^*} c_i \gamma_i - \left[ \sum_{i \in [r-t+1]} c_i (\gamma_i -\epsilon) + c_{r-t+2} (\gamma_{r-t+2} +\epsilon) + \sum_{i =r-t+3}^{s^*} c_i \gamma_i  \right] =\epsilon\left[ \sum_{i \in [r-t+1]} c_i - c_{r-t+2}\right]
\end{align*}
is greater than zero since $s^*\geq r-t+2$ and thus $\sum_{i \in [r-t+1]}c_i > c_{r-t+2} $. Therefore, setting $\epsilon= \frac{\gamma_{r-t+1} - \gamma_{r-t+2}}{2}$ leads to $\tilde{\gamma}_1 = \ldots = \tilde{\gamma}_{r-t+2}$. This process can be iterated for $s^*\geq j>r-t+2$ with $\epsilon=\frac{\tilde{\gamma}_{j-1} - \tilde{\gamma}_{j}}{j- r+t}$, i.e., set $\gamma'_i=\tilde{\gamma}_i-\epsilon$ for $i \leq j-1$ and set $\gamma_j'=\tilde{\gamma}_j+\epsilon (j - r+t-1)$. This step preserves validity since 
\begin{align*}
    &\sum_{i = r-t+1}^{j-1} (\tilde{\gamma}_i - \epsilon) + \tilde{\gamma}_j + (j-r+t-1) \epsilon = \sum_{i = r-t+1}^{j} \tilde{\gamma}_i
\end{align*}
Moreover, setting $\epsilon=\frac{\tilde{\gamma}_{j-1} - \tilde{\gamma}_{j}}{j - r+t}$ ensures that $\gamma'_i=\gamma'_j$ for $i < j\leq s^*$, since 
\begin{align*}
&\gamma'_{j-1}-\gamma'_j=\tilde{\gamma}_{j-1} - \frac{\tilde{\gamma}_{j-1} - \tilde{\gamma}_{j}}{j - r+t} - \left[\tilde{\gamma}_j + (j - r+t-1)\frac{\tilde{\gamma}_{j-1} - \tilde{\gamma}_{j}}{j - r+t}\right] = \tilde{\gamma}_{j-1} - \tilde{\gamma}_j - (\tilde{\gamma}_{j-1} - \tilde{\gamma}_{j}) = 0 
\end{align*}

Therefore, we end up with $\gamma_i = \frac{1}{s^*-r+t}$ for all $i \in [s^*]$, and $\gamma_j=0$ for all $s^*+1\leq j \leq r$. Since the objective function can not be increased further by changing the values of $\gamma$, the Lemma is proven.
\end{proof}

\begin{theorem}\label{theo:optimal_parameters}
    Given an $L$-level partitioned access structure $\mathcal{A}$, Algorithm~\ref{alg:fragments_distribution_multi-level} finds the optimal parameters $m$ and $k$ for an $[m,k]$-linear monotone erasure code for $\mathcal{A}$.
\end{theorem}
\begin{proof} 

The proof is structured as follows. We first show that the solutions provided by the algorithm for each subproblem corresponds to the optimal solution given by Lemma~\ref{lemma:bins}. This is done by induction over the depth of the tree.
We then conclude the proof by noticing that the algorithm combines subproblem solutions optimally.%

By Lemma \ref{lemma:cost_preserving}, we know that if a solution to~(\ref{LPP:to_solve}) is optimal, then all induced subproblems must also be solved optimally. Therefore, if those subproblems further decompose, their subproblems must also be solved optimally. In other words, every optimal solution to~(\ref{LPP:to_solve}) induces optimal solutions of all subproblems at every level of the decomposition. Therefore, we first show that Algorithm~\ref{alg:fragments_distribution_multi-level} optimally solves each subproblem of depth $\ell$ where $0\leq \ell\leq L-1$. We will show this by induction over $\ell$.

\textbf{Base cases:} 
$\ell=0$: Consider a subtree $T'$ consisting of a single leaf node. Note that the minimization problem of $T'$ is 
$\frac{\nu_{T'}}{k_{T'}}=\min\limits_{\gamma_1} \gamma_1$ with constraint $\gamma_1\geq 1$. Thus, the optimal solution is $\gamma_1=1$. This solution corresponds to the solution given by Algorithm~\ref{alg:fragments_distribution_multi-level}, as it returns $\nu_{T'}=1$ and $k_{T'}=1$.
Therefore, the algorithm solves every subtree problem of depth zero optimally.

Case $\ell=1$: Assume a subtree $T'$ has depth one and its root is labeled $t$ with $r$ subtrees $T_1,\ldots, T_r$ (leaves). Since an optimal solution to $T'$ must induce optimal solutions of its subproblems by Lemma~\ref{lemma:cost_preserving}, the algorithm must correspond to the solution given by the following problem 
\begin{align*}\label{subproblem:L=1}
    \frac{\nu_{T'}}{k_{T'}}=\min_{\gamma}  &\sum_{i=1}^r \gamma_i \frac{\nu_{i}}{k_{i}}  \\
     &\sum_{j\in J}\gamma_j\geq 1 \qquad \text{ for all } J\subseteq[r] \text{ with } |J|=t \tag{1}
\end{align*}
where $\gamma_i\in \BR_{\geq 0}$ and $(\nu_{i},k_{i})$ is given by the optimal solution for subproblem instantiated on $T_i$. Note that $\frac{\nu_{i}}{k_{i}}=1$. 
Therefore, if $r\neq t$, an optimal solution to problem~(\ref{subproblem:L=1}) is given by $\gamma_i=\frac{1}{t}$ for all $i\in [r]$, leading to $(\nu_{T'}, k_{T'})=(r, t)$. 
Otherwise, if $r=t$, then an optimal solution is $\gamma_i=1$ for some $i\in [r]$ and $\gamma_j=0$ for all $i\neq j\in [r]$.
Therefore, $\nu_{T'}=1$ and $k_{T'}=1$. 
Note that in both cases this corresponds to the solution achieved by Algorithm~\ref{alg:fragments_distribution_multi-level}. Indeed, since $\frac{\nu_{i}}{k_{i}}=1$, we have $$\frac{1}{s-r+t}\sum_{j=1}^{s} 1=\frac{s}{s-r+t}.$$ 
Therefore, for all $s\in \{r-t+1,\ldots,r-1\}$, $1\geq\frac{s}{s-r+t}$ only if $r=t$. Thus, if $r\neq t$, we have $s^*=r$, implying $\lambda=1$ and $\alpha_i=1$ for all $i\in [r]$. 
Then, the solution given by the algorithm is also $\nu_{T'}=r$ and $ k_{T'}=t$.
If $r=t$, then $s^*=1$ and hence $\lambda=1, \alpha_1=1$ and $\alpha_j=0$ for all $ j\in [r]\setminus\{1\}$. Thus, $(\nu_{T'}, k_{T'})=(1,1)$, which corresponds to the optimal solution.

\textbf{Induction hypothesis:} For all depths $\Delta$ such that $2\leq \Delta \leq \ell-2$, the following holds. For all subtrees $T'$ with depth $\Delta$, labeled $t$ and with $r$ subtrees $T_1,\ldots,T_r$ (depth $\Delta-1$ subtrees), the algorithm optimally solves the following problem  
\begin{align*}
    \frac{\nu_{T'}}{k_{T'}}=\min_{\gamma} &\sum_{i=1}^r \gamma_i \frac{\nu_{i}}{k_{i}} \\
    &\sum_{j\in J}\gamma_j\geq 1 \qquad \text{ for all } J\subseteq[r] \text{ with } |J|=t
\end{align*}
where $\gamma_i\in \BR_{\geq 0}$ and $(\nu_{i},k_{i})$ is given by the optimal solution for subproblem instantiated on subtree~$T_i$.

\textbf{Induction step:} Consider a subtree $T'$ with depth $\Delta+1$, labeled $t$ with $r$ subtrees $T_1,\ldots,T_r$, all with depth $\Delta$. By induction hypothesis, the algorithm delivers $(\nu_{i},k_{i})$ given by the optimal solution to each subproblem instantiated on subtree~$T_i$. Since an optimal solution to $T'$ induces optimal solutions to each subproblem $P_{i}$, we have to consider the following problem:
\begin{align*}
    \frac{\nu_{T'}}{k_{T'}}=\min_{\gamma} &\sum_{i=1}^r \gamma_i \frac{\nu_{i}}{k_{i}} \\
    &\sum_{j\in J}\gamma_j\geq 1 \qquad \text{ for all } J\subseteq[r] \text{ with } |J|=t
\end{align*}
where $\gamma_i\in \BR_{\geq 0}$. Moreover, by Lemma~\ref{lemma:bins}, we have $\frac{\nu_{T'}}{k_{T'}}=\sum_{i=1}^{s^*} \frac{\nu_{i}}{(s^*-r+t)\cdot k_{i}}$. Therefore, it holds that $k_{T'}=(s^*-r+t)\cdot \operatorname{lcm}(\{\,k_{i} \mid i\in S\,\})$ and $\nu_{T'}=\sum_{i=1}^{s^*}\nu_{i}\cdot \omega_i$ where $\omega_i=\frac{\operatorname{lcm}(\{\,k_{i} \mid i\in S\,\})}{k_{i}}$. This corresponds exactly to the solution given by Algorithm~\ref{alg:fragments_distribution_multi-level}, proving that the algorithm solves each induced subproblem optimally. 

As a last step, note that due to the same reasoning, the algorithm combines the subtree solutions of $T$ to an optimal solution, concluding the proof.
\end{proof}

\section{General Asynchronous Verifiable Information Dispersal}
\label{sec:gavid}
In this section, we propose the novel notion of \emph{general asynchronous
verifiable information dispersal (GAVID)} and implement it efficiently
based on monotone erasure codes.  GAVID generalizes the asynchronous verifiable
information dispersal (AVID) problem, as introduced by Cachin and
Tessaro~\cite{DBLP:conf/srds/CachinT05}, from threshold Byzantine quorum
systems to general Byzantine quorum systems.
Our implementation also generalizes the algorithms of Cachin and Tessaro.

AVID and GAVID are protocols run by servers that communicate with each other
by sending messages to each other over an asynchronous network.  We assume
that every two servers are linked by an asynchronous reliable point-to-point
channel, in the sense that all messages are eventually delivered to the
recipient and that the message originator is authenticated.

\subsection{Byzantine quorum systems and reliable systems}

To specify which servers may fail, we assume the presence of an adversary
capable of corrupting servers.  A \emph{fail-prone
  system}~$\CF \subseteq 2^\CP$ specifies the sets of servers that can be
jointly corrupted or be faulty during an execution (\CF is maximal, in the sense that none
of its members is contained in another).  Non-faulty servers are called \emph{honest}.
In particular, we consider a
\emph{Byzantine quorum system}~\cite{DBLP:journals/dc/MalkhiR98} \CQ for a set
of $n$ servers \CP and the corresponding \emph{canonical fail-prone system}
$\CF = \{\CP \setminus Q \mid Q \in \CQ\}$.  Then, a quorum system \CQ for \CF on
\CP is a set of subsets of \CP, called \emph{quorums}, satisfying
\emph{consistency}, that the intersection of any two quorums contains at least
one process that is not faulty, i.e.,
\[
  \forall F \in \CF, \forall Q_1, Q_2 \in \CQ : Q_1 \cap Q_2 \not\subseteq F;
\]
and \emph{availability}, namely that for any set of processes that may fail
together, there exists a disjoint quorum, that is,
\[
  \forall F \in \CF : \exists Q \in \CQ : F \cap Q \neq \emptyset.
\]

Given a Byzantine quorum system \CQ, a set \( K \subseteq \CP \) that
intersects every quorum is called a
\emph{kernel} of \CQ, that is, when
\[
  \forall\, Q' \in \CQ: K \cap Q' \neq \emptyset
\]
and $K$ is minimal among all such sets, i.e.,
\[
  \forall\, K' \subsetneq K: \exists\, Q' \in \CQ:\; Q' \cap K' = \emptyset.
\] 
The \emph{kernel system} \CK of \CQ is defined as the set of all kernels
of~\CQ.  It follows from the consistency property that for every quorum
$Q \in \mathcal{Q}$, there exists a kernel $K\in \mathcal{K}$ such that
$K\subseteq Q$: every quorum contains a kernel.

Additionally, we introduce the notion of a \emph{reliable} set, which will
play a central role in our scheme.
\begin{definition}
  A set $R \subseteq \CP$ is said to be \emph{\reliable} for a fail-prone
  system~\CF whenever one can remove an arbitrary faulty set from~$R$ and the
  difference still contains a kernel, i.e.,
  \[
    \forall F \in \mathcal{F}, \exists K\in \CK: K \subseteq R \setminus F.
  \]
  In other words, a \reliable set contains a kernel that is disjoint any
  faulty set in~\CF.
\end{definition}
The collection of all \reliable sets is called the \emph{\reliable system} of
the fail-prone system.

Note that each quorum is a reliable set, but not all reliable sets are
quorums.  To see this, just consider a Byzantine quorum system on $n$ servers
for a fail-prone system that consists of all sets of $f$ servers for
$f \ll n/3$, e.g., when $f=1$ and $n = 10$.

\subsection{Definition of general asynchronous dispersal}\label{ssec:GAVID_def}

For extending protocols that involve erasure codes from threshold to general
quorum systems, we first note that a quorum system, a kernel system, and a
reliable system is each also an access structure in the sense of
Definition~\ref{def:access.structure}.  A GAVID scheme for a quorum system \CQ
uses a suitable kernel system~\CK as access structure for its monotone erasure
code.

As in the existing notion of AVID, a GAVID scheme for \CQ consists of two
protocols, called \textsc{Disperse} and \textsc{Retrieve}, and relies on a
monotone erasure code whose access structure is \CK, the kernel system of~\CQ.
If a server~$\hat{p}$ wants to store a file~$f$, it starts protocol
\textsc{Disperse} with input~$f$, and the other servers start the protocol
with no input.  Every honest server should terminate this protocol and output
either \strf{stored} or \strf{abort}.

We say that a server \emph{completes} the dispersal if it terminates
\textsc{Disperse} and outputs \strf{stored} and that it \emph{aborts} the
dispersal if it terminates and outputs \strf{abort}.  To ensure storing
consistent data, we require that either all honest servers complete the
dispersal or no honest server completes the dispersal.  This way, we enable
any kernel (and consequently also every quorum) to reconstruct the file.

If a server $\tilde{p}$ wants to retrieve a file, it starts \textsc{Retrieve}
to obtain enough information from the servers to reconstruct the stored file.
Server~$\tilde{p}$ terminates this protocol when it outputs a file or~$\bot$.

More precisely, we define a GAVID scheme for \CQ as follows.

\begin{definition}
  A \emph{general asynchronous verifiable information dispersal} (GAVID)
  scheme for a Byzantine quorum system \CQ and a fail-prone system \CF
  consists of a dispersal protocol and a retrieval protocol.  For any
  adversary that corrupts a set $F \in \CF$ of servers and for any server
  \(\hat{p}\) that starts the dispersal, the following conditions hold:
  \begin{description}
  \item[\textbf{Termination.}] If \(\hat{p}\) is honest, then all honest
    servers eventually complete the dispersal.

  \item[\textbf{Agreement.}] If some honest server completes the dispersal,
    then all honest servers eventually complete the dispersal.

  \item[\textbf{Availability.}] If a kernel of honest servers completes the
    dispersal, and an honest server starts the retrieval protocol, then it
    eventually reconstructs some file~\(f'\).

  \item[\textbf{Correctness.}] If a kernel of honest servers completes the
    dispersal, then there exists a fixed value \(f_0\) such that:
    \begin{enumerate}
    \item If \(\hat{p}\) is honest and has dispersed a file \(f\), then
      \(f_0 = f\).
    \item If an honest server reconstructs \(f'\), then \(f_0 = f'\).
    \end{enumerate}
  \end{description}
\end{definition}

\subsection{Protocol~\textsf{GAVID}}

In the following, we present Protocol~\textsf{GAVID} that implements GAVID. We assume a collision-resistant hash function $H$ and denote the encoding and decoding functions of the monotone erasure code for $\mathcal{K}$ by \textsc{Encode} and \textsc{Decode}, respectively.

Protocol \textsf{GAVID} consists of two sub-protocols, called
\textsc{Disperse} and \textsc{Retrieve}.
Algorithm \textsc{Disperse} is instantiated by a server $\hat{p}$ wanting to disperse a file $f$. The first step is to encode $f$ using \textsc{Encode}, obtaining fragments $(g_1, \dots, g_n)$. Then, $\hat{p}$ builds the \emph{verification vector} $D= (H(g_1), \dots, H(g_n))$. 
At this  point, $\hat{p}$ sends $g_i$ along with $D$ to $p_i$, for all $i \in [n]$. Each node checks that the received fragment $g_i$ is valid with respect to $D$, i.e., $H(g_i)= D_i$, and, if so, it echoes $g_i$. Once a node $p_i$ receives valid $\strf{echo}$ messages from a quorum, it decodes them using \textsc{Decode} and encodes the result again to check if it is valid with respect to $D$. In that case, it sends a $\strf{ready}$ message containing a valid fragment $\bar g_i$ to each $p_j$. Otherwise, it means that the sender $\hat{p}$ acted maliciously, so the node aborts.

If a node does not receive valid $\strf{echo}$ messages from a quorum, but receives valid $\strf{ready}$ messages from a kernel, it can recover its fragment by decoding the received fragments and encoding again, as above.
Finally, when a node receives valid $\strf{ready}$ messages from a \reliable set $R[D]$,
it stores $D$ and its valid fragment.

When a server $\tilde{p}$ wants to retrieve a file, it runs \textsc{Retrieve}. Here, $\tilde{p}$ asks the servers to send their stored information, namely, the vector $D$ and their fragments. Once $\tilde{p}$ has received enough fragments that are valid with respect to the same $D$, it decodes them to recover the file.

The details of \textsc{Disperse} and \textsc{Retrieve} can be found in Algorithm~\ref{alg:disperse} and Algorithm~\ref{alg:retrieve}, respectively.

\begin{algo*}[htbp]
\caption{Protocol~\textsc{Disperse} (server~$p_i$)}
  \label{alg:disperse} 
\setcounter{g@linecounter}{0}

\vbox{
  \small
  \begin{numbertabbing}
    xxxx\=xxxx\=xxxx\=xxxx\=xxxx\=xxxx\=MMMMMMMMMMMMMMMMMMM\=\kill
\textbf{initialization} \label{} \\
\> $E \gets []$: hash map from $\{0,1\}^*$ to sets of nodes \label{}\\
\> \`// $E[D]$ is the set of nodes that have sent valid $\strf{echo}$ messages for $D$ \label{}\\
\> $R \gets []$: hash map from $\{0,1\}^*$ to sets of nodes \label{}\\
\> \`// $R[D]$ is the set of nodes that have sent valid $\strf{ready}$ messages for $D$ \label{}\\
\> $A \gets []$: hash map from $\{0,1\}^*$ to sets of node-fragment pairs \label{}\\
\> \`// $A[D]$ contains nodes and valid fragments for $D$ \label{}\\
\> $\var{stored} \gets \false$ \label{}\\
\label{}\\
\textbf{upon} starting dispersal with input a file~$f$ \textbf{do} \`// only server $\hat{p}$ \label{}\\
\> $[g_1, \dots, g_n] \gets \text{\textsc{Encode}}(f)$\label{}\\
\> $D \gets [H(g_1), \ldots, H(g_n)]$\label{}\\
\>\textbf{for} $j \in [n]$ \textbf{do} \label{}\\
\>\>send message $(\strf{send}, D, g_j)$ to $p_j$ \label{}\\
\label{}\\
\textbf{upon} receiving message $(\strf{send}, D, g_i)$ from $\hat{p}$ for the first time \textbf{do} \label{} \\
\>\textbf{if} $H(g_i) = D_i$ \textbf{then} \label{}\\
\>\>\textbf{for} $j \in [n]$ \textbf{do}\label{}\\
\>\>\>send message $(\strf{echo}, D, g_i)$ to $p_j$ \label{}\\
\label{}\\
\textbf{upon} receiving message $(\strf{echo}, D, g_j)$ from $p_j$ for the first time \textbf{do} \label{} \\
\>\textbf{if} $H(g_j) = D_j$ \textbf{then} \label{}\\
\>\>$A[D] \gets A[D] \cup \{(j,g_j)\}$ \label{}\\
\>\>$E[D] \gets E[D] \cup \{j\}$ \label{}\\
\>\>\textbf{if} $E[D] \in \CQ \land R[D] \not\in \CK$ \textbf{then} \`// \strf{ready} message not yet sent \label{}\\
\>\>\>$\bar{f} \gets \text{\textsc{Decode}}( (g_1, \dots, g_n) )$,
  where $(l, g_l) \in A[D]$ if $l\in E[D]\cup R[D]$ and $g_l =\bot$ otherw.\label{}\\ %
\>\>\>$[\bar{g}_1, \dots, \bar{g}_n]\gets \text{\textsc{Encode}}(\bar f)$\label{}\\
\>\>\>\textbf{if} $H(\bar g_l) = D_l$ for all $l \in [n]$  \textbf{then}\label{}\\
\>\>\>\>\textbf{for} $j \in [n]$ \textbf{do}\label{}\\
\>\>\>\>\>send message $(\strf{ready}, D, \bar g_i)$ to $p_j$\label{}\\
\>\>\>\textbf{else}\label{}\\
\>\>\>\>output $\strf{abort}$\label{}\\
\label{}\\
\textbf{upon} receiving message $(\strf{ready}, D, g_j)$ from $p_j$ for the first time \textbf{do} \label{}\\
\>\textbf{if} $H(g_j) = D_j$ \textbf{then} \label{}\\
\>\>$A[D] \gets A[D] \cup \{(j,g_j)\}$ \label{}\\
\>\>$R[D] \gets R[D] \cup \{j\}$ \label{}\\
\>\>\textbf{if} $E[D] \notin \CQ \land R[D] \in \CK$ \textbf{then} \`// \strf{ready} message not yet sent \label{}\\ 
\>\>\>$\bar{f} \gets \text{\textsc{Decode}}( (g_1, \dots, g_n) )$, where $(l, g_l) \in A[D]$ if $l\in E[D]\cup R[D]$ and $g_l =\bot$ otherw.\label{}\\ 
\>\>\>$[\bar{g}_1, \dots, \bar{g}_n] \gets \text{\textsc{Encode}}(\bar f)$ \label{}\\
\>\>\>\textbf{if} $H(\bar g_l) = D_l$ for all $l \in [n]$ \textbf{then} \label{}\\
\>\>\>\>\textbf{for} $j \in [n]$ \textbf{do}\label{}\\
\>\>\>\>\>send message $(\strf{ready}, D, \bar g_i)$ to $p_j$\label{}\\
\>\>\>\textbf{else}\label{}\\
\>\>\>\>output $\strf{abort}$\label{}\\
\>\>\textbf{else if} $\neg \var{stored} \land R[D] \in \CR$ \textbf{then} \label{}\\
\>\>\>store $(D, g_i)$ \label{} \\
\>\>\>$\var{stored} \gets \true$ \label{}\\
\>\>\>output $\strf{stored}$ \label{}
  \end{numbertabbing}
}
\end{algo*}

\begin{algo*}[htbp]
\caption{Protocol~\textsc{Retrieve} (server $p_i$)}\label{alg:retrieve}
\setcounter{g@linecounter}{0}

\vbox{
  \small
  \begin{numbertabbing}
    xxxx\=xxxx\=xxxx\=xxxx\=xxxx\=xxxx\=MMMMMMMMMMMMMMMMMMM\=\kill
\textbf{initialization} \label{} \\
\> $(D, g_i) \gets \text{values retrieved from \textsc{Disperse}}$ \label{}\\
\> $R \gets []$: hash map from $\{0,1\}^*$ to sets of nodes \label{}\\
\> \`// $R[D]$ is the set of nodes that have sent valid $\strf{fragment}$ messages for $D$ \label{}\\
\> $A \gets []$: hash map from $\{0,1\}^*$ to sets of node-fragment pairs \label{}\\
\> \`// $A[D]$ contains nodes and valid fragments for $D$ \label{}\\
\> $\bar{f} \gets \bot$ \`// the retrieved file \label{}\\
\label{}\\
\textbf{upon} starting retrieval \textbf{do} \`// only server $\tilde{p}$ \label{}\\
\>\textbf{for} $j \in [n]$ \textbf{do} \label{}\\
\>\>send message $(\strf{retrieve})$ to $p_j$ \label{}\\
\label{}\\
\textbf{upon} receiving message $(\strf{retrieve})$ from $p_{j}$ \textbf{do} \label{} \\
\>\textbf{if} $D \neq \bot$ \textbf{then} \label{}\\
\>\>send message $(\strf{fragment}, D, g_i)$ to $p_{j}$ \label{}\\
\label{}\\
\textbf{upon} receiving message $(\strf{fragment}, D, g_j)$ from $p_j$ for the first time \textbf{and} $\bar{f} = \bot$ \textbf{do} \label{}\\
\>\textbf{if} $H(g_j)=D_j$ \textbf{do} \label{}\\
\>\>$A[D] \gets A[D] \cup \{(j,g_j)\}$ \label{}\\
\>\>$R[D] \gets R[D] \cup \{j\}$ \label{}\\
\>\>\textbf{if} $R[D] \in \CK$ \textbf{then} \label{}\\
\>\>\>$\bar{f} \gets \text{\textsc{Decode}}( (g_1, \dots, g_n) )$,
  where $(l, g_l) \in A[D]$ for $l \in [n]$ and $g_l=\bot$ otherwise\label{}\\
\>\>\>output $\bar{f}$\label{}
 \end{numbertabbing}
  }
\end{algo*}
\subsection{Analysis of Protocol~\textsf{GAVID}}\label{ssec:GAVID_analysis}
In this section, we prove that, given an $[m,k]-$linear monotone erasure code for a quorum system \CQ, \textsf{GAVID} provides a general asynchronous verifiable information dispersal scheme for \CQ. We show this statement in Theorem~\ref{theo:GAVID}, building on the following lemma.

\begin{lemma} \label{Lemma: D^i=D^j}
    In the dispersal protocol of \textsf{GAVID}, if two honest servers $p_i$ and $p_j$ send $(\strf{ready}, D^i)$ and $(\strf{ready}, D^j)$, respectively, then $D^i=D^j$.
\end{lemma}
\begin{proof}
    Assume by contradiction that $D^i\neq D^j$. We know that $p_i$ and $p_j$ send $\strf{ready}$ messages once they have received either $\strf{echo}$ messages from a quorum, or $\strf{ready}$ messages from a kernel. Then, we can have the following three cases.
    
    Case 1: Assume both $p_i$ and $p_j$ receive $\strf{echo}$ messages from a quorum. Then since any two quorums intersect, at least one honest server $p_r$ has sent two different $\strf{echo}$ messages to $p_i$ and $p_j$, which gives a contradiction.

    Case 2: Assume one of $p_i$ and $p_j$ has received $\strf{echo}$ messages from a quorum and the other has received $\strf{ready}$ messages from a kernel. W.l.o.g, assume $p_i$ has received $(\strf{echo}, D^i)$ messages from a quorum and $p_j$ has received $(\strf{ready}, D^j)$ messages from a kernel. Then one honest server must have sent $(\strf{ready}, D^j)$ from a quorum. But then it reduces to case 1.

    Case 3: Assume both $p_i$ and $p_j$ have received $(\strf{ready}, D^i)$, resp. $(\strf{ready}, D^j)$ messages from a kernel. Then at least one honest server $p_r$ must have sent a $(\strf{ready}, D^i)$, and at least one honest server $p_l$ must have sent $(\strf{ready}, D^j)$. W.l.o.g., we can assume that both $p_r$ and $p_l$ have received $(\strf{echo}, D^i)$, resp. $(\strf{echo}, D^j)$ messages from a quorum. But then it reduces again to case 1.  
\end{proof}

\begin{theorem}\label{theo:GAVID}
    Given an $[m,k]-$linear monotone erasure code for a quorum system \CQ, \textsf{GAVID} is a general asynchronous verifiable information dispersal scheme for \CQ. %
\end{theorem}

\begin{proof}
\textbf{Termination.} %
If $\hat{p}$ is honest and starts \textsc{Disperse} with $D$ then all honest servers will receive $(\strf{send},D)$ together with its fragments. Therefore, each honest server will send $(\strf{echo}, D)$ and receive echo messages from a quorum.  Thus, each honest server will send $(\strf{ready},D)$. To see this, assume there is an honest server $p_j$ that aborts and an honest server $p_i$ that sends $(\strf{ready},D)$. Therefore, there must exist a vector $\overline{g}$ such that $H(\overline{g}_\ell)=D_\ell$ for all $\ell\in [n]$. However, since $p_j$ aborts, by Lemma \ref{Lemma: D^i=D^j} there exists $(\ell,g'_\ell)\in A^j[D]$ such that $H(g'_\ell)=D_\ell$ and $g'_\ell\neq \overline{g}_\ell$. 
This means that the adversary has found a collision for $H$, which is impossible. Therefore, at the end each honest server will receive valid fragments with the same $D$ from a \reliable set $R[D]$. Thus, each honest server will terminate with $D$.

\textbf{Agreement.} Assume one honest server has completed the dispersal protocol and has stored the verification vector $D$. This means that it has received valid fragments with the same $D$ from a \reliable set $R[D]$. 
Therefore, some kernel $K$ of honest servers has sent valid ready messages with $D$, %
which means that each honest server receives valid ready messages with $D$ from $K$. %
This implies that all honest servers send a valid ready message with $D$. This can be seen as follows. Assume a server $p_i$ completes the dispersal with $D$ and there is a server $p_j$ that aborts. This means that there exists a value $\overline{g}$ such that $H(\overline{g}_\ell)=D_\ell$ for all $\ell\in [n]$. However, since $p_j$ aborts and by Lemma \ref{Lemma: D^i=D^j}, there exists $(\ell,g'_\ell)\in A^j[D]$ such that $H(g'_\ell)=D_\ell$, but $g'_\ell\neq \overline{g}_\ell$, i.e., the adversary has found a collision for $H.$  %
Therefore, in the end, each honest server receives valid ready messages from a \reliable set $R'[D]$. %
We prove this fact by contradiction. 
Suppose that for all $R'[D] \subseteq \CP$ there exists $ F \in \CF$ such that $ \nexists K\in \CK: K \subseteq R'[D] \setminus F$. From the availability property of Byzantine quorum systems, it follows that there is a quorum $Q_F$ such that $Q_F \cap F = \emptyset$. %
Since every quorum contains a kernel, there exists $K_F \in \CK$  such that $K_F \subseteq Q_F$. Thus, $K_F \cap F = \emptyset$. We can therefore set $R'[D] = \bigcup_{F \in \mathcal{\CF}} K_F$, obtaining a contradiction. 

\textbf{Availability.} Since a kernel $K$ of honest servers has accepted, servers in $K$ also hold the same verification vector
$D$, and thus any honest server is always able to reconstruct some value. 

\textbf{Correctness.} Let $R[D] \subseteq \CP$ be a kernel of honest servers that have completed the dispersal of~$f$. Define $f_0 = \text{\textsc{Decode}}(\overline{g}_{R[D]})$.
For the first point, we have to show that $f_0$ corresponds to the file $f$ a honest server has started to disperse. Assume an honest server $\hat{p}$ has shared a file $f$ and $f \neq f_0$. Then every $\strf{echo}$ message from an honest $p_i$ to an honest $p_j$ contains $D$ and $g_i$ as computed by $\hat{p}$. If the servers in $R[D]$ computed their $\overline{g}_j$ from these $\strf{echo}$ messages, then $g_j = \overline{g}_j$, i.e., their fragments correspond to the one sent by $\hat{p}$.
However, since $f \neq f_0$, there must be an honest server $p_j$ with $\overline{g}_j \neq g_j$. Thus, it must
have received a value $g'_\ell \neq g_\ell$ from a corrupted server $p_\ell$ (either in an $\strf{echo}$ or $\strf{ready}$ message),
which it accepted. Since $H(g_\ell) = H(g'_\ell)$ holds, the adversary has found a collision for $H$, which is impossible.

For the second point, we have to show that if an honest server reconstructs a value $f'$, then this value corresponds to $f_0$. For this,  assume an honest server reconstructs a value $f' \neq f_0$ using some kernel $R'[D] \neq R[D]$.  Since $R'[D]$ is a kernel, the value of $D$ the server chooses must be the unique
one held by the correct servers by Lemma~\ref{Lemma: D^i=D^j}. On the other hand, if $f' \neq f_0$, there must be some
value $g'_\ell$ received by some server $p_\ell$, with $\ell \in R'[D] \setminus R[D]$, such that $g'_\ell \neq \overline{g}_\ell$, but $H(g'_\ell) = D_\ell$.
On the other hand, in order for a server
$p_j \in R[D]$ to accept, it must hold that $H(\overline{g}_\ell) = D_\ell$.
Thus, the adversary must have found a collision for $H$ either in the dispersal protocol or in the retrieval protocol.
\end{proof}

\textbf{Complexity.}  
We use the following complexity measures to analyze \textsf{GAVID}: 
\begin{itemize}
    \item \emph{message complexity}: number of messages associated to an instance of the protocol;
    \item \emph{communication complexity}: bit length of all messages associated to an instance of the protocol;
    \item \emph{storage complexity}: overall bit length of the information stored in the memory of the honest servers after they have completed the dispersal protocol.
\end{itemize} 
Let us denote by $|H|$ the bit-size of the range of the hash function. %
The message complexity of \textsc{Disperse} is $\mathcal{O}(n^2)$. 
Regarding the communication complexity, we know that the bit length of each message is $m_i \cdot \lceil log_2(q) \rceil + n \cdot |H|$, for some $i \in [n]$. Since each party sends the same message $n$ times, %
we have that the communication complexity is 
\begin{align*}
    & \mathcal{O} \Big(\sum_{i \in [n]} (n \cdot m_i\cdot \lceil log_2(q) \rceil + n^2 |H|) \Big) =  \\
    & \mathcal{O} \Big(n \cdot m\cdot \lceil log_2(q) \rceil + n^3 |H|) \Big) = \\
    & \mathcal{O}(n\cdot m\cdot \frac{|f|}{k}+n^3\cdot |H|) =\\
    & \mathcal{O}(n(1+\beta)\cdot |f|+n^3\cdot |H|)
\end{align*}
where $m=k(1+\beta)$ and $|f|=\mathcal{O}(k\cdot \lceil log_2(q)\rceil )$.
The storage complexity is given by the size of all stored fragments, which is $m \cdot \lceil log_2(q) \rceil $, together with the size of vector $D$ taken $n$ times, i.e., $n^2 \cdot |H|$. The resulting complexity is $\mathcal{O} (m \cdot \lceil log_2(q) \rceil + n^2 \cdot |H|) = \mathcal{O} ((1+\beta) \cdot |f| + n^2 \cdot |H|)$. %

The message complexity of \textsc{Retrieve} is $\mathcal{O}(n)$.
For the same reasoning as above, the communication complexity is $\mathcal{O} ((1+\beta) \cdot |f| + n^2 \cdot |H|)$.

\textbf{\textsf{GAVID-H}.} The storage and communication complexities of the protocol can be improved using a Merkle tree instead of the vector $D$, following the approach of Cachin and Tessaro in~\cite{DBLP:conf/srds/CachinT05}. In particular, the idea is to build the Merkle tree whose leaves are
labeled with the hashes of the fragments $H(g_1), \dots, H(g_n)$. Let $l= \lceil \log(n) \rceil$, and denote by $v_0, \dots, v_l$ the unique path from the root $v_r$ to leaf $i$. Assign $label(v_i)=H(label(v'), label(v''))$ to each vertex $v_i$ with children $v'$ and $v''$. 
Moreover, let $w_j$ be the unique sibling of $v_j$ for all $j \in [l]$.
Then, for each $i \in [n]$, the \emph{fingerprint} for $g_i$ is $FP(i)= [label(w_1), \dots, label(w_l)]$.

\textsf{GAVID} can then be modified as follows. Replace $D$ with the root hash $H_r$ in all messages, and let server $p_i$ store $(FP(i), H_r)$ instead of $D$. Moreover, the check $H(g_i)=D_i$ is substituted with the following test. Let $\textsf{verify}(i, g_i, FP, H_r)$ be the function that computes the iterated hash on $H(g_i)$ and $FP$, with $FP$ a vector of length $l$, and then compares the result with $H_r$. If they are the same, the function returns 1, otherwise it returns 0.
We call the resulting protocol \textsf{GAVID-H}. 

We now show that \textsf{GAVID-H} is more efficient \textsf{GAVID}. First, notice that the bit length of each message now is $m_i \cdot \lceil log_2(q) \rceil + (\lceil log_2(n)\rceil + 1) \cdot |H|$. For the same reasoning as above, this leads to the following communication complexities:
\begin{itemize}
    \item $\mathcal{O}(n(1+\beta)\cdot |f|+n^2(\lceil log_2(n)\rceil + 1)\cdot |H|)$ bits in \textsc{Disperse};
    \item $\mathcal{O} ((1+\beta)\cdot |f| + n \cdot (\lceil log_2(n)\rceil + 1)\cdot |H|))$ bits in \textsc{Retrieve}.
\end{itemize}
Moreover, the storage complexity of \textsc{Disperse} is $\mathcal{O} ((1+\beta)\cdot |f| + n \cdot (\lceil log_2(n)\rceil + 1)\cdot |H|))$.

\subsection{Communication-efficient general Byzantine reliable broadcast}

Byzantine reliable broadcast is a fundamental asynchronous communication
primitive, which exists as a building block also in many consensus protocols.
It ensures that a sender can transmit a message to all servers such that if
one correct server delivers some message, then all correct servers deliver a
message.  Moreover, they all deliver the same message, and if the sender is
correct, then this is the message that the sender has broadcast.

Bracha's~\cite{DBLP:journals/iandc/Bracha87} celebrated protocol implements
this notion in an asynchronous network with an initial message from the sender
to all servers followed by two rounds of all-to-all communication.

As already observed by Cachin and Tessaro~\cite{DBLP:conf/srds/CachinT05},
their notion of AVID and the corresponding protocol are closely related to the
notion of Byzantine reliable broadcast and Bracha's protocol.  Hence,
communication-efficient reliable broadcast for a Byzantine quorum system~\CQ
can be implemented with a protocol that is closely related to the dispersal
algorithm of Protocol~\textsf{GAVID} for~\CQ.  Relating to the description of
\textsf{GAVID} in Algorithm~\ref{alg:disperse}, the broadcast protocol
involves the same \strf{send} message from the sender to all servers, and it
runs almost the same two rounds of exchanging \strf{echo} and \strf{ready}
messages.  The only difference occurs in the last step: once a server has
received valid fragments from a \reliable set $R[D]$, it immediately decodes
the valid fragments and outputs the resulting message.  This method extends to
related communication-efficient broadcast protocols that use erasure
codes~\cite{DBLP:conf/wdag/Shoup24, DBLP:conf/eurocrypt/LocherS25}.

\section{Related work}
\label{sec:related_work}

We first review the most commonly used erasure codes. We then discuss storage methods that leverage different properties of erasure codes and present data availability sampling, another key application of erasure codes. Additionally, we mention broadcast and consensus protocols based on erasure coding techniques.
Finally, we discuss the possibility of building monotone erasure codes for access structures expressed as \emph{monotone span programs} (MSPs).

\textbf{Erasure codes.}
One of the most important classes of erasure codes is given by \emph{maximum distance separable} (MDS) codes, as they provide an optimal trade-off between fault tolerance and storage overhead. The most widely used MDS codes are \emph{Reed–Solomon} (RS) codes~\cite{ReedS60, DBLP:journals/jacm/Rabin89}, due to their simplicity and efficiency. Whenever a node fails and needs to recover its fragment, traditional MDS codes repair it by downloading the whole file and encoding it again.
To address this issue, Dimakis et al.~\cite{DBLP:journals/tit/DimakisGWWR10} introduced \emph{regenerating codes} and characterized the fundamental tradeoff between storage and repair bandwidth. Among them, \emph{minimum-storage regenerating} (MSR) codes minimize repair bandwidth while maintaining the minimum storage overhead. In particular, they achieve the same storage optimality as Reed–Solomon codes and are therefore MDS.

\textbf{Distributed storage.} 
Erasure codes are typically employed to efficiently store large amount of information in a distributed manner.

\emph{Redundant Array of Independent Disks} (RAID)~\cite{DBLP:conf/sigmod/PattersonGK88} is an erasure coding technique involving multiple disk arrays, organized in different ways. The most widely used variants are RAID-5 and RAID-6, which give single-disk and double-disk fault tolerance, respectively.

Filecoin~\cite{DBLP:conf/dsn/PsarasD20} is a \emph{decentralized storage network} (DSN) scheme that offers a data storage and retrieval service through a decentralized network of independent storage providers. 
A flaw of Filecoin is given by slow reads in the absence of a node storing a hot copy of the file. To address this issue, Williams et al. propose Arweave~\cite{Arweave}. Their approach, however, leads to replication levels comparable to classic state machine replication, implying a very high overhead. 
A much more efficient decentralized storage system is Storj~\cite{Storj}, which achieves an overhead of 1.75. %
Nevertheless, a drawback of Storj is its inefficient data recovery process, which requires reconstructing the entire file in order to recover lost parts. 
A more flexible scheme is Walrus~\cite{DBLP:journals/corr/abs-2505-05370}, which provides an efficient reconstruction mechanism while ensuring Byzantine Fault Tolerance. However, Walrus is less efficient in terms of overhead, which is 3.5. %
The state-of-the-art protocol for distributed storage is Shelby~\cite{DBLP:journals/corr/abs-2506-19233}, which achieves an overhead of less than 1, %
lower than that of all other distributed storage systems present in the literature. Shelby employs Clay codes~\cite{DBLP:conf/fast/VajhaRPKLSKBYNH18}, a particular kind of MSR codes, to provide \emph{hot storage}, i.e, an infrastructure that can support low-latency, high-throughput reads. 

All the distributed storage systems presented here assume a $t$-out-of-$n$ access structure, which makes them less flexible than the scheme proposed in this paper.

\textbf{Data availability sampling.} Another relevant application of erasure codes is \emph{data availability sampling} (DAS). 
In the context of blockchains, DAS allows light clients to verify that the data of a block has been fully published without downloading the entire block. 
The idea is that the block is first encoded using erasure codes (typically Reed–Solomon codes). The encoded data is then committed using Merkle trees. Light clients randomly sample fragments and verify them against the commitment. If enough randomly selected fragments are available, the client can conclude with high probability that the entire block is available. Blockchains such as Celestia~\cite{Celestia} are built around DAS as a core component. Moreover, a prominent data availability layer is EigenDA~\cite{EigenDA}, which is designed for Ethereum rollups.

\textbf{Communication-efficient broadcast and consensus.} Many recent
Byzantine-tolerant consensus protocols disseminate data with erasure
codes~\cite{DBLP:journals/csur/YangCWWLZ24}. %
Indeed, data dissemination based on erasure coding can be leveraged to improve the efficiency of multiple protocols. In particular, such techniques can be used to realize fast validated agreement~\cite{DBLP:conf/wdag/CivitDGGKVZ24, DBLP:conf/dsn/0001LM025}.

Moreover, erasure coding techniques can be leveraged to build efficient broadcast protocols. For example, Shoup~\cite{DBLP:conf/wdag/Shoup24} builds on AVID to reduce the communication complexity of the Simplex atomic broadcast protocol~\cite{DBLP:conf/tcc/ChanP23}. 

\textbf{Erasure coding using Monotone Span Programs.}
For threshold access structures, erasure coding schemes can be constructed directly from polynomials. The idea is to embed the information into the coefficients of a degree-$t$ polynomial. Any set of at least $t+1$ nodes can later reconstruct the polynomial, thus recovering the information.
At the same time, it is known that \emph{monotone span programs} (MSPs) are a generalization of polynomials~\cite{DBLP:conf/coco/KarchmerW93}~-- as a matter of fact, access structures can be represented by MSPs~\cite{Beimel96, DBLP:conf/coco/KarchmerW93}, and secret-sharing schemes have been described using MSPs~\cite{DBLP:conf/eurocrypt/CramerDM00,Gennaro96,DBLP:conf/sss/AlposC23}. Although it is also possible to build erasure codes from MSPs, doing so would not be efficient. For this reason, in this paper we take a different approach and do not rely on MSPs.

\section{Conclusion}
\label{sec:conclusion}

In this work, we have introduced monotone erasure codes that respect arbitrary trust assumptions characterized by a monotone access structure. They extend existing erasure codes, which rely on threshold failure models, to applications with more general assumptions, where some nodes are more likely to fail than others or may be more trustworthy. Practical systems that rely on such assumptions can be found in the blockchain and cryptocurrency space.

We have shown how to efficiently build a linear monotone erasure code for any given access structure~$\CA$, without minimizing the overhead.
To generate a linear monotone erasure code that is optimal, we have presented an alternative method that uses a threshold erasure code as a building block.
An optimal code of this kind results from solving a linear programming problem derived from the access structure. The procedure is efficient if \CA has a compact (e.g., polynomial-size) description in the number of nodes. We have also provided an efficient construction of monotone erasure codes for partitioned access structures.

Monotone erasure codes have also direct applications to Byzantine-tolerant communication protocols: We have extended AVID so that it works not only in scenarios where any $f$ out of $n$ nodes may fail, but also in the general case of Byzantine quorum systems.  The resulting GAVID primitive has led us also to obtain a communication-efficient Byzantine reliable broadcast protocol for an arbitrary Byzantine quorum system.

This work points to future work in multiple directions.
Perhaps the most interesting open question concerns the efficient construction of an optimal monotone erasure code for an arbitrary monotone access structure, which is expressed as MBF or monotone span program, but that has exponentially many access sets. 
Furthermore, constructing \textit{dynamic} monotone erasure codes that can adapt to access structures that evolve over time would be valuable in practical applications.
Finally, it would be interesting to determine to which extent other methods from coding theory, such as rateless codes or low-density parity-check codes, can be employed in the setting of monotone erasure codes.

\section*{Acknowledgments}

V.B. and A.C. have been supported by a grant from the Stellar Development
Foundation to the University of Bern.

\printbibliography

\end{document}